\documentclass[conference,compsoc]{IEEEtran}
\ifCLASSOPTIONcompsoc
  \usepackage[nocompress]{cite}
\else
  \usepackage{cite}
\fi
\ifCLASSINFOpdf
  \usepackage[pdftex]{graphicx}
  \graphicspath{{./figures/}}
  \DeclareGraphicsExtensions{.pdf}
\else
\fi
\usepackage{array}
\newcolumntype{?}{!{\vrule width 1.3pt}}

\usepackage{url}

\usepackage[dvipsnames]{xcolor}
\usepackage{xurl}
\usepackage{booktabs}
\usepackage{multirow}
\usepackage{xspace}
\usepackage{makecell}
\usepackage{tablefootnote}
\usepackage{pifont}

\newenvironment{packeditemize}{
\begin{list}{$\bullet$}{
\setlength{\itemsep}{1.5pt}
\setlength{\labelwidth}{8pt}
\setlength{\leftmargin}{10pt}
\setlength{\labelsep}{3pt}
\setlength{\listparindent}{\parindent}
\setlength{\parsep}{1.5pt}
\setlength{\parskip}{1.5pt}
\setlength{\topsep}{1.5pt}}}{\end{list}}

\newcommand{\para}[1]{\noindent{\bf #1}}

\usepackage{algorithm}
\usepackage{algpseudocode}

\usepackage[normalem]{ulem}

\usepackage{amsthm}
\usepackage{amsmath}
\usepackage{amssymb}
\newtheorem{theorem}{Theorem}
\newtheorem{definition}{Definition}
\newtheorem{lemma}{Lemma}

\usepackage[hidelinks]{hyperref}
\hypersetup{colorlinks=false}

\begin{document}
\emergencystretch 30pt
%
\title{On the Effectiveness of NetFlow-Based Deanonymization of Tor Onion Services (Extended Version)}

\author{\IEEEauthorblockN{Jinseo Lee}
\IEEEauthorblockA{KAIST\\
jinseo.vik.lee@kaist.ac.kr}
\and
\IEEEauthorblockN{Jongkook Han}
\IEEEauthorblockA{KAIST\\
jkookhan03@kaist.ac.kr}
\and
\IEEEauthorblockN{Yixin Sun}
\IEEEauthorblockA{University of Virginia\\
ys3kz@virginia.edu}
\and
\IEEEauthorblockN{Min Suk Kang$^{\dagger}$}
\IEEEauthorblockA{KAIST\\
minsukk@kaist.ac.kr}}


%


\maketitle
\thispagestyle{plain}
\pagestyle{plain}

\makeatletter
    \begingroup
        \renewcommand{\thefootnote}{\fnsymbol{footnote}}
        \def\@IEEEcompsocmakefnmark{\hbox{\normalfont\@thefnmark\ }}
        \footnotetext[2]{Corresponding author.}
    \endgroup
\makeatother

\begin{abstract}
In the context of Tor, deanonymization attacks via traffic correlation have long been studied for both clients and onion services.
Complementing prior investigations that utilize fine-grained flow metadata such as packet sizes and timings, we study a relatively less explored yet easier-to-obtain adversarial capability: deanonymization using only coarse NetFlow-style telemetry, such as one-minute, 1-out-of-1,000 sampled packet counts.
We first revisit active NetFlow-based traffic analysis attacks against Tor onion services and show, through large-scale experiments with live Tor data, that coarse telemetry is sufficient to deanonymize a targeted onion service when the adversary can induce discernible traffic patterns through controlled downloads.
We then introduce a stronger active attack based on controlled introduction requests, which removes application-specific assumptions and gives adversaries more flexibility while achieving comparable success rates.
Finally, we present a passive deanonymization attack that sends zero traffic to target onion services.
The attack exploits Tor's puzzle-based DoS defense as a side channel: public puzzle difficulty reveals service-load changes that can be semantically matched against NetFlow traces.
Our findings show that coarse network telemetry poses a practical deanonymization threat to onion services today and call for proactive mitigations and a redesign of the globally visible defense signals that may leak service-load information.
\end{abstract}


%
\IEEEpeerreviewmaketitle

\section{Introduction}
\label{sec:intro}
Deanonymization is one of the most serious threats to the Tor ecosystem~\cite{murdoch05low,evans09practical,bauer07low,johnson13users,sun15raptor,nasr18deepcorr,oh22deepcoffea,overlier06locating,biryukov13trawling,jansen14sniper,lopes24flow,murdoch07sampled,chakravarty14effectiveness}.
Since its initial design~\cite{dingledine04tor}, Tor has accepted an inherent tradeoff between anonymity and latency: it provides practical low-latency communication but does not aim to protect against global adversaries capable of observing traffic at large scale.
This design choice rests on the assumption that collecting traffic metadata from sufficiently many network vantage points is prohibitively costly for most adversaries.

Against this backdrop, one relatively overlooked yet operationally realistic adversary capability
is access to coarser-grained network telemetry, such as minute-level sampled packet counts from NetFlow.
Such telemetry is routinely collected in operational networks for monitoring, traffic engineering, anomaly detection, and volumetric attack mitigation~\cite{steenbergen05advanced,berthier07netflow,boschi07from,filsfils11best,freedman15many}.
Although far noisier than packet-level traces, it may become a powerful deanonymization signal when collected at scale (e.g., by regional ISPs or a state-level adversary).

Prior work has shown that this threat is plausible.
In 2007, Murdoch and Zieli\'{n}ski~\cite{murdoch07sampled} suggested, through simulation results and analytic modeling, that IXPs could observe a significant portion of Tor traffic and thereby correlate two traffic traces: one entering the Tor network and another exiting it.
Later in 2014, Chakravarty et al.~\cite{chakravarty14effectiveness} demonstrated that NetFlow telemetry can deanonymize Tor clients in a decoy-service setting by inducing a recognizable traffic pattern at an adversary-controlled decoy server and searching for the same pattern near candidate clients.
Although this attack demonstrated the feasibility of NetFlow-based deanonymization, it relies on a specific assumption: the adversary must lure a target client to the decoy service and keep the target engaged long enough to induce a traffic signal.
The Tor Project also raised concerns about the attack's practicality and scalability, citing false positive rates and the base-rate fallacy when searching among many candidates~\cite{dingledine14traffic}.

To the best of our knowledge, there has been little academic follow-up or public discussion on the real-world impact of NetFlow-based deanonymization after Chakravarty et al.'s work and the Tor Project's response in 2014.
In this work, we therefore revisit the largely overlooked adversary capability, i.e., large-scale coarse-grained telemetry, but study it against Tor onion services rather than Tor clients.

This shift in interest, from Tor clients to onion services, matters because deanonymization of onion services poses a serious threat to the global anonymity-service landscape.
Onion services are often operated under severe surveillance, censorship, or legal risks~\cite{lopes24flow,maurizi26welcome,securedrop26directory}, and they protect content providers as well as content consumers.
Consequently, exposing an onion service can reveal the operator of a sensitive service, with potentially greater consequences than deanonymizing a Tor client.
Furthermore, onion services are more vulnerable to this specific class of attacks because anyone can initiate communication with them, as we show in this paper.

Building on existing knowledge and practice regarding the privacy of Tor onion services, we ask how effective NetFlow-based deanonymization attacks are against onion services.
We answer this question by studying attacks across a spectrum: from active traffic induction through normal service access and introduction requests to passive deanonymization that requires no traffic induction at all.
Our findings are threefold.
First, NetFlow-only adversaries can deanonymize onion services by inducing download traffic, without the social engineering (e.g., a decoy web service) required by prior client-side attacks.
Second, application-level downloads are unnecessary: introduction requests alone induce NetFlow-visible traffic patterns that do not depend on target-specific content and are comparably effective for deanonymization.
Third, deanonymization can even be passive: a new side channel in Tor's Denial-of-Service (DoS) mitigation protocol exposes coarse-grained DoS-related signals, which can be correlated with large-scale NetFlow telemetry to deanonymize onion services under DoS conditions.

Taken together, these findings show that coarse-grained operational telemetry, previously considered too noisy or limited for reliable deanonymization, constitutes a practical anonymity attack capability against Tor onion services.
They show that NetFlow-only adversaries are not limited to carefully staged decoy-service attacks against Tor clients, but can also conduct generic attacks that exploit in-protocol features against Tor onion services.

The root cause of these attacks is access to large-scale, coarse-grained telemetry.
However, restricting access to such telemetry is challenging: legal or regulatory limits may help in some jurisdictions, but they cannot provide a global defense.
We thus argue that technical mitigations should assume that large-scale NetFlow telemetry is accessible to sufficiently capable adversaries, and suggest several mitigation directions: coarse-grained traffic-signature detection, periodic IP rotation, and changes to Tor's DoS mitigation protocol to reduce passive side-channel leakage.

Our contributions are summarized as follows:
\begin{packeditemize}
    \item We revisit NetFlow-based deanonymization against a new target, Tor onion services, and show that onion services are more easily deanonymized than Tor clients studied in prior work,
    \item We design and evaluate two active NetFlow-only attacks: download induction and introduction induction,
    \item We identify a new side channel in Tor's DoS mitigation protocol that enables even passive deanonymization,
    \item We empirically evaluate all the proposed attacks using ethically collected Tor traffic summaries and traces from attack executions involving our own servers, and
    \item We also evaluate several potential mitigation strategies and propose promising directions for mitigation.
\end{packeditemize}
\section{Background and Related Work}
\label{sec:background}

\subsection{Tor and Onion Services}
\label{sec:onion-service-background}
Tor~\cite{dingledine04tor} is an overlay network designed to provide anonymity for low-latency, TCP-based applications.
It consists of globally distributed relay servers operated by volunteers.
A Tor client constructs a path, or circuit, through multiple relays to reach the intended destination.
Each relay can identify only its immediate predecessor and successor and does not know the complete path.
A typical Tor circuit involves three relays: an entry (guard) node, a middle node, and an exit node.
When a client accesses the public Internet through this circuit, the traffic leaves the Tor network via the exit node and is forwarded to the destination server.
Because the destination server observes only the identity of the exit node, the client's identity remains hidden; i.e., sender anonymity is preserved.

The Tor onion service protocol~\cite{tor26rendezvous} extends anonymity protection to the server side, concealing both the client's identity and the server's location.
Unlike conventional Internet services that use standard domain names, onion services employ a unique addressing scheme that ends with \texttt{.onion}.
These services are hosted entirely within the Tor network, ensuring that all communication remains internal to Tor.
To establish a connection, an onion service selects a set of relays known as introduction points, through which clients can initiate contact, and publishes this information in its descriptor through directory servers.
The client then selects a relay called a rendezvous point, where the actual communication takes place, and sends an introduction request containing rendezvous information to the onion service via one of the introduction points.
By communicating through the rendezvous point without revealing the server's location, the onion service can provide its service while remaining hidden; i.e., receiver anonymity is preserved.

\newcommand{\yes}{\color{Green}\ding{51}}
\newcommand{\no}{\color{purple}\ding{56}}

\begin{table*}[tb]
    \setlength{\tabcolsep}{4pt}
    \renewcommand{\arraystretch}{1.12} 
    \caption{Qualitative Comparison of Threat Models for Traffic Analysis Attacks on Tor}
    \label{tab:threat-comp}
    \footnotesize
    \centering
    \begin{tabular}{l?c|c|c|c|c|c|c?c|c?c}
        \toprule
        & \textbf{\cite{murdoch05low,evans09practical}} & \textbf{\cite{bauer07low}} & \textbf{\cite{johnson13users}} & \textbf{\cite{sun15raptor}} & \textbf{\cite{nasr18deepcorr,oh22deepcoffea}} & \textbf{\cite{overlier06locating,zhang11application,biryukov13trawling,ling13protocol,jansen14sniper}} & \textbf{\cite{lopes24flow}} & \textbf{\cite{murdoch07sampled}} & \textbf{\cite{chakravarty14effectiveness}} & \textbf{Ours} \\
        \midrule
        \textbf{Relay Adversary} & {\yes} & {\yes} & {\yes} & {\no} & {\yes} & {\yes} & {\no} & {\no} & {\no} & {\no} \\
        \textbf{Network Adversary} & {\no} & {\no} & {\yes} & {\yes} & {\yes} & {\no} & {\yes} & {\yes} & {\yes} & {\yes} \\
        \textbf{Coarse Metadata (e.g., NetFlow) Only} & {\no} & {\no} & {\no} & {\no} & {\no} & {\no} & {\no} & {\yes} & {\yes} & {\yes} \\
        \textbf{Direct Interaction with Target} & {\yes} & {\no} & {\no} & {\no} & {\no} & {\yes} & {\no} & {\no} & {\yes} &  {\no}/{\yes} \\
        \textbf{Protocol Bug Exploited} & {\no} & {\yes} & {\no} & {\no} & {\no} & {\yes} & {\no} & {\no} & {\no} & {\no} \\
        \textbf{Target$^{*}$} & C & C & C & C & C & O & B & C & C & O \\
        \bottomrule
        \multicolumn{11}{l}{\footnotesize $^{*}$C denotes Tor clients, O denotes onion services, and B denotes both.}
    \end{tabular}
\end{table*}

\subsection{Traffic Analysis for Deanonymization}
\label{sec:deanon-background}
In low-latency anonymity networks, including Tor, traffic metadata (e.g., packet volumes and timings) is known to be difficult to conceal, making them vulnerable to traffic analysis attacks~\cite{serjantov03passive,danezis04traffic,levine04timing}.
Consequently, studies have proposed various traffic analysis attacks against Tor clients and onion services.

Against {\em Tor clients}, Murdoch and Danezis~\cite{murdoch05low} showed in 2005 that an adversary could estimate the traffic load on specific relays and use this information to identify Tor clients visiting adversarial destinations (e.g., an adversary-controlled decoy server); Evans et al.~\cite{evans09practical} later revisited and improved the attack in 2009.
In 2007, Bauer et al.~\cite{bauer07low} proposed a low-cost traffic analysis attack based on relay bandwidth inflation.
Johnson et al.~\cite{johnson13users} evaluated the practical risk posed by various traffic analysis adversaries in 2013, while Sun et al.~\cite{sun15raptor} investigated the impact of Internet routing on traffic analysis attacks in 2015.
Nasr et al.~\cite{nasr18deepcorr} and Oh et al.~\cite{oh22deepcoffea} proposed machine learning-based traffic analysis attacks in 2018 and 2022, respectively.

Against {\em Tor onion services}, {\O}verlier and Syverson~\cite{overlier06locating} demonstrated in 2006 that an adversary could deanonymize a Tor onion service by operating a single relay.
They showed that an adversary could determine whether her relay was part of the target's rendezvous circuit by analyzing traffic metadata (i.e., the timing and number of Tor cells), and eventually deanonymize the target onion service by repeatedly initiating rendezvous connections.
This attack led to the widespread adoption of entry guards~\cite{wright03defending},\footnote{With entry guards, each client or onion service selects guard relays from a limited set rather than from the entire relay set. As a result, repeatedly establishing new connections does not increase the probability of selecting an adversarial relay.} mitigating the risk of such repeated probing.
In 2011, Zhang et al.~\cite{zhang11application} showed that an adversary could deanonymize a Tor onion service by counting cells, assuming that the adversary controlled the target's entry guard.
Later in 2013, Biryukov et al.~\cite{biryukov13trawling} revealed flaws in the onion service directory protocol that enabled opportunistic deanonymization attacks based on traffic analysis.
That same year, Ling et al.~\cite{ling13protocol} demonstrated that an adversary could quickly determine whether her relay served as the first node in a target's rendezvous circuit.
A year later, Jansen et al.~\cite{jansen14sniper} showed how DoS attacks can be used to facilitate the deanonymization of onion services.
Specifically, they demonstrated that an adversary can position one of its relays as a guard node for a target onion service by identifying the target's current guard relays and launching DoS attacks against them.
By repeatedly disrupting the target's guards, the adversary can eventually force the target to select a new guard relay under the adversary's control.

More recently, Lopes et al.~\cite{lopes24flow} introduced machine learning-based traffic analysis attacks targeting onion services and their clients.
They showed that, by using dense and fine-grained traffic metadata, an adversary can (1) distinguish onion service traffic from client traffic, (2) differentiate clearnet client traffic from onion service client traffic, and (3) correlate onion service client traffic with its corresponding onion service traffic.

\subsection{{\em NetFlow}-Based Deanonymization Attacks}
\label{sec:netflow-based-background}
Most effective traffic analysis attacks against Tor we discuss above use the accurate packet-level traffic metadata (e.g., packet size, timing) of Tor flows to perform correlation-based deanonymization. 
Another possible, yet less explored, approach to deanonymization of Tor onion services is to use much coarser-grained traffic metadata, such as sampled and aggregated flow records.
A typical example of such coarse-grained traffic metadata is NetFlow data, which have been widely deployed for traffic engineering and network management since the early 2000s~\cite{steenbergen05advanced,berthier07netflow,boschi07from,filsfils11best,freedman15many}.
They are easier to obtain and analyze than the accurate traffic metadata because they are {\em sampled} (e.g., only 1 out of every 1,000 packets is recorded) and {\em aggregated} (e.g., only the sum of sampled packet features is recorded over fixed time intervals, such as one minute).
Thus, the amount of collected traffic metadata is several orders of magnitude smaller than that of accurate traffic metadata (e.g., millions of packet-level records can be reduced to a single integer). 

Murdoch and Zieli\'{n}ski~\cite{murdoch07sampled}, for the first time in the public literature, proposed a NetFlow-based traffic correlation attack against Tor clients in 2007.
Their work suggested that NetFlow-based deanonymization of Tor clients could be feasible, and their claim was supported by simulation results and analytical modeling.

Chakravarty et al.~\cite{chakravarty14effectiveness} later extended this earlier work by evaluating the attack on the live Tor network using actual NetFlow data collected from a router.
Not only did they evaluate the attack in a real-world setting, but they also proposed a practical attack strategy for inducing identifiable traffic flows on target Tor clients.
One limitation of their approach, however, is that it requires a complicated setup for traffic induction, namely, tricking a target client into visiting an adversary-controlled decoy website and downloading a large file for at least 5--7 minutes.

Following this line of work, the Tor Project officially acknowledged the basic feasibility of these NetFlow-based attacks~\cite{murdoch07sampled,chakravarty14effectiveness} in its 2014 official statement~\cite{dingledine14traffic}.
Yet, in the same statement, the Tor Project cited the reported false positive rates of 6.6\% (6 false positives in 90 trials, each involving around 900 candidate flows) and 12.5\% (3 in 24 trials, each involving around 1,500 candidate flows) as evidence that the attack would {\em not} scale well to real-world deployments, where the number of candidate flows can reach millions; see the official rebuttal for further details~\cite{dingledine14traffic}.

Table~\ref{tab:threat-comp} summarizes the qualitative comparison of threat models for fine-grained traffic analysis attacks~\cite{murdoch05low,evans09practical,bauer07low,johnson13users,sun15raptor,nasr18deepcorr,oh22deepcoffea,overlier06locating,zhang11application,biryukov13trawling,ling13protocol,jansen14sniper,lopes24flow} and NetFlow-based deanonymization attacks~\cite{murdoch07sampled,chakravarty14effectiveness} on Tor.
In this study, we consider a network adversary that has access only to coarse-grained traffic metadata and targets Tor onion services either with or without direct interaction with the targets (i.e., through active or passive attacks).
\section{Problem Definition}
\label{sec:problem-definition}
In this section, we define the threat model and problem setting for NetFlow-based deanonymization of Tor onion services.
We begin by describing the adversary's objectives and capabilities and examine the practicality of the required adversarial capabilities.
Then, we formulate the research questions that guide our exploration of active and passive deanonymization attacks against onion services.

\subsection{Threat Model}
\label{sec:threat-model}
We consider an adversary whose goal is to {\em deanonymize} Tor onion services by associating onion addresses of interest with their corresponding IP addresses.
The adversary may target either a single {\em Onion of Interest} (OoI) or a collection of OoIs (e.g., the top-N most popular services within a specific sector, such as illegal pharmaceutical markets).
In addition to a set of OoIs, the adversary maintains a set of {\em candidate IP addresses}, consisting of IP addresses of hosts that communicate with Tor relays using the Tor protocol but are not themselves listed as Tor relays.

Examples of such adversaries can be easily found in practice.
Several governments have conducted operations aimed at deanonymizing onion services, including those of Germany~\cite{marx24law}, South Korea~\cite{wi22fighting}, and the US~\cite{doj14more,tippe24onion}.
Notably, some of these operations targeted more than 400 onion services, highlighting the motivation for large-scale deanonymization attacks.

We consider a network adversary (e.g., controlling one or more ISPs) whose network carries traffic for a substantial portion of the IP network space, as would be the case for a large tier-1 ISP or medium-sized regional ISP.
The network adversary is assumed to be equipped with industry-standard monitoring capabilities that are widely deployed in commercial routers, namely the NetFlow~\cite{cisco26netflow} flow monitoring feature and the ISP's existing centralized network monitoring infrastructure (e.g., the flow collectors and analysis tools typically used in network operations centers).
From this vantage position, the adversary can passively monitor a significant number of IP addresses at a coarse-grained level.
In addition, they can identify the ones that use the Tor Protocol by observing traffic from/to many Tor relays over time.

Additionally, for active attacks, we consider that the adversary can use her own Tor clients to induce specific, coarse-grained traffic patterns through either application-level interactions or introduction requests (see \S\ref{sec:traffic-inducing} to better understand how these two differ).

\subsection{Feasibility of Large-Scale NetFlow Monitoring}
\label{sec:large-scale-arguments}
One critical adversary capability we require in this study is to obtain {\em sampled and aggregated NetFlow telemetry} from a medium- to large-sized ISP network(s).
In this section, we discuss the feasibility of such large-scale NetFlow monitoring using real-world data and examples.
We first explain how an ISP can compile the list of candidate IP addresses.
Then, we assess the practicality of collecting NetFlow telemetry for these candidate IP addresses at scale.

\subsubsection{Large-scale identification of candidate IPs}
\label{sec:large-scale-ip}
The adversary must first identify candidate IP addresses, i.e., non-relay Tor-speaking IP addresses that are observable from the monitored network(s).
This task can be carried out using infrastructure already available to ISPs, as most Tor relay IPs and ports are publicly known~\cite{austin26tor,nusenu26tor,scriptz26tor}.
An adversary could, for instance, identify candidate IPs by analyzing historical network data (e.g., firewall logs) for connections to Tor guard relays.
Or, the adversary could configure sampled NetFlow monitoring for all guard relays to determine which candidate IPs are observable.

In this work, we consider a candidate IP address to be observable if packets in at least one direction between the IP address and one of its guard relays consistently transit the adversary network.
For example, a candidate IP address is considered observable if packets sent either to or from a single guard relay pass through the adversary network.
Our evaluation confirms that such partial traffic visibility is sufficient for effective deanonymization.
We attribute the limited impact of partial visibility to the fact that our attacks rely on \textit{changes} in traffic patterns rather than on absolute traffic volumes; see the results in \S\ref{sec:download-inducing}, \S\ref{sec:intro-inducing}, and \S\ref{sec:passive-eval}.

Because candidate IPs may include both Tor clients and onion services, we consider a simple filtering technique that removes IP addresses exhibiting only short-lived, sporadic Tor activity to focus on well-established, active onion services~\cite{mani18understanding}.
Nevertheless, the resulting candidate-IP list may still include some Tor clients.
We show, using real-world Tor traffic data we collect, that the presence of these clients does not diminish attack effectiveness; see the results in \S\ref{sec:download-inducing}, \S\ref{sec:intro-inducing}, and \S\ref{sec:passive-eval}.

\subsubsection{Large-scale network telemetry monitoring}
Once the adversary has identified the candidate IP addresses, the adversary monitors the traffic volumes to/from all observable candidate IPs.
This can be achieved using {\em sampled NetFlow}~\cite{cisco26netflow} (or similar techniques, such as sFlow~\cite{sflow26about} or IPFIX~\cite{aitken13specification}), which has been available as a basic feature in commodity routers since as early as 1996~\cite{hofstede14flow}.
When commodity routers are configured to use sampled NetFlow (typically {from 1:100 to 1:5000 sampling~\cite{popa23what,petryschuk24netflow,cloudflare25recommended,fastnetmon25fastnetmon}}) for a certain set of IPs, they send only summaries of the highly sampled packets to a centralized network monitoring system.
The default (and recommended) reporting interval for these summaries is one minute~\cite{petryschuk24netflow,juniper24netflow}, which we confirm to be sufficient for our attacks.

To gauge the practical overhead of data collection in a large-scale deanonymization campaign, we derive ballpark estimates of the number of candidate IPs that an adversary needs to monitor.
Using these estimates, we evaluate whether commodity routers can support such large-scale monitoring in terms of configuration complexity, storage capacity, monitoring bandwidth, and router CPU/memory usage.
We consider an adversary operating at a major ISP (e.g., tier-1 or large tier-2 ISPs) seeking to deanonymize a substantial number of onion services.

\vspace{3pt}
\para{Number of candidate IPs.}
First, we estimate the observability of candidate IP addresses at a major ISP.
Because we lack such data directly measured from ISPs, we instead use AS-level BGP simulations with the CAIDA AS relationships dataset~\cite{luckie13as,caida25as}.

\begin{figure}[t]
    \centering
    \includegraphics[width=\linewidth]{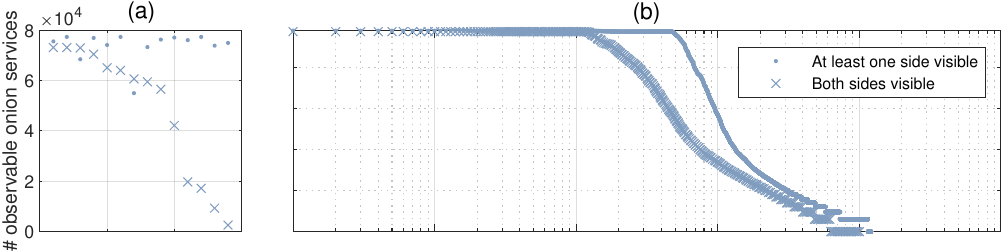}
    \caption{Number of observable onion service hosts for (a) tier-1 ISPs and (b) tier-2 ISPs, each showing two cases: both ingress and egress paths visible, and only one side visible. Note that (b) is shown on a log-log scale.}
    \label{fig:bgp}
\end{figure}

\begin{table*}[th]
    \setlength{\tabcolsep}{4pt}
    \renewcommand{\arraystretch}{1.20} 
    \caption{Router Resource Utilization and Monitoring Accuracy on a Juniper MX204$^{*}$}
    \label{tab:router}
    \footnotesize 
    \centering
    \begin{tabular}{lcccccc}
        \toprule
        \textbf{Number of Monitored IPs} & \textbf{0 (Control)} & \textbf{Top-10} & \textbf{Top-100} & \textbf{Top-1,000} & \textbf{Top-10,000} & \textbf{Top-100,000} \\
        \midrule
        \textbf{Packet Drops} & 0 & 0 & 0 & 0 & 0 & 0 \\
        \textbf{J-Flow Errors} & 0 & 0 & 0 & 0 & 0 & 0 \\
        \textbf{Memory Overload} & 0 & 0 & 0 & 0 & 0 & 0 \\
        \textbf{IPFIX Packets} & - & 35,285.8 & 101,584.3 & 160,329.3 & 246,469.8 & 309,003 \\
        \textbf{J-Flow Records} & - & 118,115 & 360,962 & 586,435.3 & 927,415.8 & 1,175,091.3 \\
        \textbf{J-Flow Total Relative Error (\%)} & - & 0.00 & -0.05 & -0.05 & 0.00 & -0.04 \\
        \textbf{FPC CPU Utilization (\%)$^{**}$} & (7.4, 8, 8.8) & (7.4, 8, 8.5) & (7.4, 8, 8.3) & (7.4, 8, 8.7) & (7.4, 8, 8.5) & (7.4, 8, 8.5) \\
        \textbf{RE CPU Utilization (\%)$^{**}$} & (12.5, 13, 16.8) & (12.5, 13, 17) & (12.4, 13, 16.8) & (12.6, 13, 17.3) & (12.4, 13, 17) & (12.4, 13, 16.3) \\
        \textbf{FPC Heap Utilization (\%)$^{**}$} & (9, 9, 9) & (9, 9, 9) & (9, 9, 9) & (9, 9, 9) & (9, 9, 9) & (12, 12, 12) \\
        \textbf{FPC Buffer Utilization (\%)$^{**}$} & (9, 9, 9) & (9, 9, 9) & (9, 9, 9) & (9, 9, 9) & (9, 9, 9) & (9, 9, 9) \\
        \textbf{RE Buffer Utilization (\%)$^{**}$} & (6.8, 7, 7) & (6.6, 7, 7) & (6.3, 6.8, 6.8) & (6.2, 6.3, 6.3) & (6.4, 6.5, 6.8) & (6.9, 7, 7) \\
        \bottomrule
        \multicolumn{7}{l}{\footnotesize $^{*}$All results are averages across four repeated experiments.\quad$^{**}$Each tuple reports mean, 95th percentile, and maximum.} \\
    \end{tabular}
\end{table*}

It has been previously studied that a large network adversary (e.g., an ISP or IXP) can observe or compromise a significant fraction of Tor traffic~\cite{johnson13users,nithyanand16measuring,rochet20claps,sun15raptor,sun17counter,wails18tempest,edman09as}, and AS-relationship-based BGP route simulation is widely used in these studies.
Building on this line of work, we estimate how many onion services would be visible to a large network adversary (i.e., how many have traffic between themselves and their guard nodes transiting through the adversary) using the CAIDA AS relationship dataset~\cite{luckie13as,caida25as} and the public Tor relay dataset~\cite{tor26onionoo}.
First, we generate an AS topology graph using the CAIDA dataset and classify ASes into three tiers: fourteen well-known large ASes\footnote{AS 701, 1299, 2914, 3257, 3320, 3356, 3491, 5511, 6453, 6461, 6762, 6830, 7018, and 12956} are designated as tier-1; ASes with no customers are labeled as stub ASes; all others are classified as tier-2.
Next, we identify ASes that host at least one Tor guard node (referred to as guard-ASes).
We assume that every AS hosts exactly one onion service, since the AS-level distribution of onion services is unknown.
Based on this assumption, we compute paths between each AS (representing an onion service) and all guard-ASes, and check whether these paths transit the adversary's AS.

When calculating paths, we simulate BGP route propagation under Gao-Rexford-style routing policies until convergence.
Specifically, routes learned from customers are exported to all neighboring ASes, whereas routes learned from peers or providers are exported only to customers, enforcing valley-free routing.
When multiple routes are available, an AS selects the route with the highest local preference, preferring customer routes over peer routes and peer routes over provider routes; remaining ties are broken by AS-path length and then by a deterministic hash of the next-hop ASN.
We repeat this process for every assumed adversary (14 times for tier-1 ISPs and 12,116 times for tier-2 ISPs) and calculate the average fraction of onion services that each can observe.
This yields an estimate of the proportion of the network that is observable by an adversary at each tier.

Figure~\ref{fig:bgp} shows the number of observable onion service hosts for each type of adversary.
When considering only one side of the path (either ingress or egress), nearly all tier-1 ISPs can monitor a significant portion of onion services: on average, a single tier-1 ISP can observe 73,774 out of 77,787 onion services (95\%).
Similarly, several hundred tier-2 ISPs have comparable monitoring capabilities; for example, the top 600 ISPs can observe an average of 69,469 onion services, or about 89\%.

To assess the maximum resource requirements, we assume that tier-1 adversaries and large tier-2 adversaries can monitor near-100\% of onion services on the Tor network, because we consider a candidate IP address to be observable if packets in at least one direction between the IP address and one of its guard relays transits the adversary network.
We emphasize that this is intended to demonstrate the technical feasibility of the attack (e.g., estimating the number of ACL rules required for sampled NetFlow configurations and the corresponding storage demands) rather than to precisely estimate the number of affected onion services.
Since the real-world distribution of onion services across ASes is not well understood, accurately quantifying the impact remains challenging.

In addition, we approximate the number of active, frequently accessed onion services as follows.
Based on estimates of the total number of unique onion services (around 600,000)~\cite{tor26onion} and the highly skewed distribution of onion service access patterns (roughly the top 5\% of onion services fall into a substantially higher-access upper tail)~\cite{wang23comprehensive}, we estimate that there are around 30,000 such major onion services: $600,000 \times 0.05 = 30,000$.

Based on the high observability of candidate IPs at large ISPs and our approximation of the number of well-established, active onion services, we estimate that an assumed adversary may monitor approximately 30,000 candidate IP addresses.

\vspace{3pt}
\para{Configuration complexity, storage capacity, and monitoring bandwidth.}
With this estimate, we now assess the cost and resource requirements of network monitoring.
The telemetry data needed for our attacks are highly coarse-grained and low rate.
They are 1-out-of-1,000 sampled packet counts~\cite{popa23what,petryschuk24netflow,cloudflare25recommended,fastnetmon25fastnetmon}, with only one report collected per minute using standard NetFlow or sFlow.
Implementing this setup requires configuring Access Control Lists (ACLs) with about 30,000 filtering rules, which typical routers readily handle; see devices like Cisco Nexus 1000V Switch~\cite{cisco26acls} (released in 2009, retail price \$41,000, supporting 1.28 million rules) or Extreme 8720-32C~\cite{router26extreme,extreme21acl} (released in 2020, retail price \$33,510, supporting 102,400 rules per ACL type).
This setup would generate 360~KB of NetFlow exports per minute, amounting to around 516~MB/day or 15.6~GB/month~\cite{cisco11format}, well within the capacity of a large ISP.

\vspace{3pt}
\para{Router CPU/memory usage.}
To confirm the feasibility at a router level, we conduct real-world experiments with a commodity router.
We use a Juniper MX204 router\footnote{The Junos OS version is 21.2R3-S8.5.}~\cite{router26juniper} (released in 2017, retail price \$40,350) and measure router resource usage during large-scale traffic monitoring.
We replay the CAIDA Anonymized Internet Traces 2019 dataset~\cite{caida19anonymized} and vary the set of monitored IP addresses while measuring router resource utilization and monitoring accuracy.
This dataset contains 2.29 billion packets in total, with an average rate of 636,000 packets per second, representing high-volume realistic traffic.
For monitoring, we configure J-Flow, Juniper's own implementation of NetFlow~\cite{juniper24netflow}, with IPFIX and vary the number of monitored IP addresses from 0 (control) to 100,000.
We use 1:1,000 sampling with a default timeout of 60 seconds (15 seconds for inactive flows) and aggregate packet counts every minute for each IP address.
For comparison across different experiments, we sort all IP addresses in the dataset by the number of recorded packets and select the top-$N$ IP addresses for each experiment, where $N$ is the number of monitored IP addresses.
For example, when monitoring 10 IP addresses, we select the top 10 IP addresses by recorded packet count.
We run the experiment four times: once while gradually increasing the size of the monitored IP set, once while gradually decreasing it, and twice with randomized orders.

Our commodity router reliably supports coarse-grained traffic monitoring at scale using its version of NetFlow (J-Flow) in the presence of high-volume traffic.
Table~\ref{tab:router} summarizes the results.
All reported results are averages across the four repeated experiments, and FPC and RE denote Flexible PIC Concentrator (i.e., data plane) and Routing Engine (control plane), respectively.
We make three notable observations.
First, across all configurations, we observe {\em no} packet drops, J-Flow monitoring errors, or memory overload.
Second, the resulting reports are highly accurate: the largest absolute relative error compared to the ground truth is only 0.05\%.
Third, both CPU and memory utilization remain stable and lower end throughout all experiments.
Together, these results show that a commercial router can reliably and accurately monitor 100,000 active IP addresses while forwarding high-volume, realistic (e.g., generating up to 309,003 IPFIX packets and 1,175,091 records) traffic.

\subsection{Our Research Questions}
\label{sec:reserach-questions}
In this study, we investigate the following main research question: ``{\em How effective are NetFlow-based deanonymization attacks against Tor onion services?}''
To answer this, we consider a set of specific research questions:
\begin{packeditemize}
    \item[\textbf{RQ1.}] ``(Revisiting prior attacks) Can we induce download traffic on an onion service(s) to deanonymize it?'' (\S\ref{sec:download-inducing})
    \item[\textbf{RQ2.}] ``(Extending active attacks) Can we leverage introduction requests to induce traffic on an onion service(s) to deanonymize it?'' (\S\ref{sec:intro-inducing})
    \item[\textbf{RQ3.}] ``(Towards passive attacks) Can we deanonymize onion services without sending any traffic to them?'' (\S\ref{sec:passive})
    \item[\textbf{RQ4.}] ``(Exploring defenses) Can we effectively mitigate NetFlow-based deanonymization attacks?'' (\S\ref{sec:mitigation})
\end{packeditemize}

By answering these questions, we provide a comprehensive assessment of the threats that NetFlow-based deanonymization adversaries pose to Tor onion services.
In the following section, we summarize the attacks studed in this work.
\section{Attack Overview}
\label{sec:attack-overview}
We introduce two categories of NetFlow-based deanonymization attacks: active and passive attacks.
Our active attack studies revisit and extend prior attack work~\cite{chakravarty14effectiveness} by exploring traffic-inducing deanonymization in the context of Tor onion services.
In contrast, our passive attack study explores a new attack vector.
Our goal is to evaluate whether an adversarial capability first studied nearly two decades ago~\cite{murdoch07sampled}, yet largely neglected for over a decade, remains effective against a different target, namely onion services rather than Tor clients.

\begin{table*}[tb]
    \setlength{\tabcolsep}{4pt}
    \renewcommand{\arraystretch}{1.20}
    \caption{Comparison of Three NetFlow-Based Onion Deanonymization Attacks}
    \label{tab:summary}
    \footnotesize 
    \centering
    \begin{tabular}{l|c|c|c}
        \toprule
        & \multicolumn{2}{c|}{\bf Active Onion Deanonymization (\S\ref{sec:traffic-inducing})} & {\bf Passive Onion Deanonymization (\S\ref{sec:passive})} \\
        \cline{2-4}
        & \textbf{Download-Inducing Attack} & \textbf{Introduction-Inducing Attack} & \textbf{Puzzle Side-Channel Attack} \\
        \midrule
        \makecell[l]{\bf Target Selection} & \multicolumn{2}{c|}{Targeted (i.e., one OoI at a time)} & \makecell{Untargeted\\(i.e., multiple OoIs simultaneously)} \\
        \hline
        \makecell[l]{\bf Application\\\bf Requirements} & \makecell{OoIs must support\\reliable downloads} & None & None \\
        \hline
        \makecell[l]{\bf Attack Persistence} & \multicolumn{2}{c|}{Short-term (e.g., hours at most)} & Long-term (e.g., months to years) \\
        \hline
        \makecell[l]{\bf Puzzle Requirements} & \multicolumn{2}{c|}{None} & OoIs must enable puzzle defenses \\
        \bottomrule        
    \end{tabular}
\end{table*}

\begin{figure}[!t]
    \centering
    \includegraphics[width=0.98\linewidth]{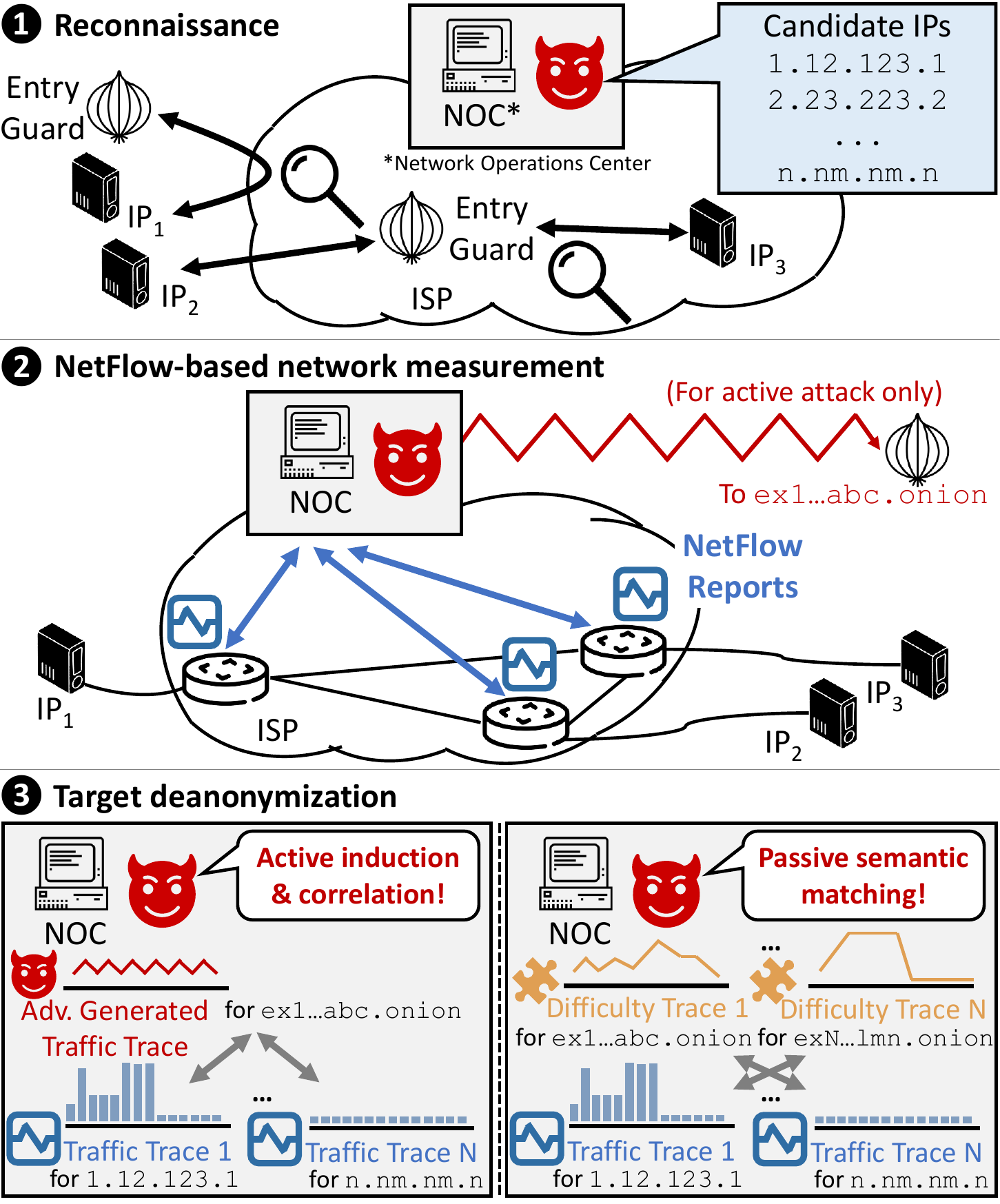}
    \caption{The three phases of our deanonymization attacks.}
    \label{fig:phases}
\end{figure}

\subsection{General Attack Phases}
\label{sec:attack-phases}
Figure~\ref{fig:phases} illustrates the three phases of our deanonymization attacks, sequentially launched by a network adversary that controls a network operations center (NOC) of an ISP: \ding{202}~reconnaissance, \ding{203}~NetFlow-based network measurement, and \ding{204}~target deanonymization.

\vspace{3pt}
\para{Phase-\ding{202}: Reconnaissance.}
The adversary first chooses a list of \texttt{.onion} addresses of interest (Onions of Interest, or OoIs) as attack targets.
And the adversary obtains a list of IP addresses that are candidate network addresses of the chosen OoI(s).
In the active attacks, each attack campaign targets one OoI and attempts to identify one of the candidate IPs that is most likely to be the network address of the targeted OoI.
Multiple active attack campaigns can be coordinated and executed in parallel or subsequently depending on the adversary's final goal.
In passive attacks, on the contrary, a single attack campaign attempts to deanonymize all OoIs in the list, trying to find matched IPs from the candidate list.
Further details on how each attack deanonymizes the OoIs are provided in the subsequent sections: \S\ref{sec:traffic-inducing} and \S\ref{sec:passive}.

\vspace{3pt}
\para{Phase-\ding{203}: NetFlow-based network measurement.}
Next, the adversary collects network telemetry data using sampled NetFlow.
For example, the adversary gathers one-minute aggregated packet counts (referred to as {\em traffic traces}) for each candidate IP identified during the reconnaissance phase, with a sampling rate of 1-out-of-1,000, which is a widely recommended sampling rate from multiple vendors~\cite{popa23what,petryschuk24netflow,cloudflare25recommended,fastnetmon25fastnetmon}. 
In addition to the passive NetFlow-based measurement, in the case of active attacks, the adversary also induces traffic on the OoI and records one-minute aggregated packet counts of the traffic she generates.

\vspace{3pt}
\para{Phase-\ding{204}: Target deanonymization.}
Finally, the adversary attempts to deanonymize the OoI(s).
As the exact procedure differs substantially between active and passive attacks, we defer the details to \S\ref{sec:traffic-inducing} and \S\ref{sec:passive}, respectively.

\subsection{Summary of Our Attacks}
\label{sec:summary}
Table~\ref{tab:summary} compares the three NetFlow-based deanonymization attacks at a high level.
The two active attacks are targeted, require adversary-induced traffic, and can be launched on demand over short time periods.
In contrast, the passive attack is untargeted (i.e., it can deanonymize multiple OoIs simultaneously), operates without interacting with its targets, and is opportunistic in nature.
These attacks present different trade-offs in terms of scalability, operational requirements, and attack opportunities.
\section{{\em Active} NetFlow-Based Deanonymization}
\label{sec:traffic-inducing}
In this section, we discuss two active NetFlow-based deanonymization attacks that induce certain traffic patterns on target Onions of Interest (OoIs) through different mechanisms.
First, we present an attack that induces traffic patterns through rate-controlled, reliable downloads (\S\ref{sec:download-inducing}), which is an extension of an approach proposed in a prior study~\cite{chakravarty14effectiveness}.
Second, we introduce a new active attack that leverages introduction requests for active traffic induction (\S\ref{sec:intro-inducing}).

\subsection{Download-Inducing Deanonymization}
\label{sec:download-inducing}
The first active attack extends earlier work~\cite{chakravarty14effectiveness} presented in 2014.
Chakravarty et al.~\cite{chakravarty14effectiveness} empirically evaluated the effectiveness of NetFlow-based deanonymization attacks against Tor clients (not onion services though) using application-level traffic patterns.
Extending this earlier approach, we design a simple download traffic inducing (or download-inducing) attack against Tor onion services.
The core idea behind this attack is to induce specific, coarse-grained traffic patterns (e.g., square-wave patterns) on the OoI by downloading files from it.
Therefore, the attack requires the adversary to be able to reliably maintain downloads while controlling the download rate.
By flexibly controlling the download rate, the adversary can generate desired traffic patterns on the OoI.

When showing the feasibility of this attack, in this study, we set up the OoI to host a sufficiently large file that can be downloaded for the desirable duration (e.g., 1 minute) and the adversary repeatedly initiates and aborts downloads of this file.
In practice, even if the onion service does not host such a large file, the adversary could synchronize multiple downloads of smaller files to generate similar high-amplitude traffic patterns.

\vspace{3pt}
\para{Three attack phases.}
The adversary begins with Phase-\ding{202} reconnaissance, where the adversary chooses a list of OoIs and their candidate IP addresses.
In the default attack mode (which we mainly describe in this paper), the adversary targets one OoI at a time.
Multiplexing of concurrent active attacks against multiple OoIs is possible with distinct traffic patterns (or codes) with low inter-code correlation at the cost of reduced attack accuracy.
Unless otherwise noted, we assume the default single-target attack mode in our analysis and evaluation.
We provide additional experiment results on multiplexing concurrent attacks at the end of this section.

During Phase-\ding{203}, while collecting one-minute aggregated packet counts for candidate IPs, the adversary also {\em induces} a certain traffic pattern (e.g., a square-wave pattern) on the target OoI for some time (e.g., one hour in our experiments) and records the resulting traffic.
For example, to create a square-wave pattern, the adversary alternates between downloading a large file from the OoI for 60 seconds and remaining idle for 60 seconds.

In Phase-\ding{204} target deanonymization, the adversary correlates the download traffic she successfully generates with the coarse NetFlow measurements for many candidate IPs.
The adversary computes the Pearson correlation coefficient between the adversary-generated traffic trace (i.e., the induced traffic pattern) and each observed traffic trace (i.e., the traffic trace associated with each candidate IP).
The trace pair with the highest correlation coefficient is identified as a candidate match.
The adversary may accept the pair as a true match only if the coefficient exceeds a certain threshold for a higher confidence level.

This relatively simple attack strategy offers a clear improvement over the similar active attacks against Tor clients proposed by Chakravarty et al.~\cite{chakravarty14effectiveness}.
Unlike those attacks, it can be launched without tricking a target into visiting a decoy website.

\vspace{3pt}
\para{Evaluation I (small-scale, isolated).}
We demonstrate the feasibility and effectiveness of download-inducing attacks at a small scale and in an isolated setting, but on the {\em live} Tor network.
We deploy an actual onion service within our institution and on AWS,\footnote{All Tor implementations are based on version 0.4.8.12-dev (93df26b).} and carefully conduct the attacks against it.
Our adversary runs a single Tor daemon with an attack script that periodically initiates downloads from the OoI.
To obtain non-target onion service traffic traces, we conduct a separate experiment with another onion service that is never targeted by the adversary.
This evaluation is isolated in the sense that it considers only traffic traces generated within our experiment setup.
Later in this section, we incorporate real-world, unrelated traffic traces to evaluate the attack in a more open setting.

We conduct two 26-hour experiments under otherwise identical conditions.
In the first, the adversary targets our OoI by repeatedly downloading a 10~GB HTML file for 60 seconds, aborting the download, and then remaining idle for 60 seconds, effectively creating a square-wave pattern that alternates every minute.
In the second, the adversary remains completely idle, allowing us to collect non-attack traffic traces (i.e., non-target onion service traffic traces).
Throughout both experiments, we record one-minute aggregated packet counts near the onion service at a sampling rate of 1-out-of-1,000.
During the first experiment, we also record one-minute aggregated packet counts at the adversary's own machine (without sampling) which serves as the adversary-generated traffic trace.
We repeat the experiments twice, for a total of three runs, to collect additional traffic traces.

We introduce realistic background traffic into the experiments.
We deploy five scripted clients that continuously access our onion service using Tampermonkey~\cite{biniok26home}.
In addition to the large file used for the attack, the onion service hosts the Wikipedia archive~\cite{kiwix26wikipedia} using Zim and the scripted clients access this Wikipedia archive.
We consider the resulting background traffic to be highly noisy, thus, enabling us to test the attacks in a harsh environment.
First, five scripted clients repeatedly click random internal links at uniformly distributed intervals of 5--60 seconds, producing an average aggregate rate of 9.2 request per minute.
For context, a prior study examined 39,824 onion services and observed request rates for 3,140 of them~\cite{biryukov14content}; out rate falls between those of the 18th- and 19th-ranked services, indicating heavy background traffic.
Furthermore, the served content consists of highly diverse web pages (8.4 million pages averaging 44~MB, with a standard deviation of 66~MB), causing the background traffic volume to vary considerably across requests.

\begin{table}[t]
    \setlength{\tabcolsep}{4pt}
    \renewcommand{\arraystretch}{1.20} 
    \caption{Evaluation Results of Download-Inducing Attack}
    \label{tab:download-inducing}
    \footnotesize 
    \centering
    \begin{tabular}{l|c|c}
        \toprule
        & \textbf{Small-Scale (Isolated)} & \textbf{Large-Scale (Open)} \\
        \midrule
        \makecell[l]{\textbf{Attack Setting}} & 1-of-11 & 1-of-12,041 \\
        \hline
        \textbf{Success Rate} & 100\% (97$/$97) & 100\% (97$/$97) \\
        \hline
        \textbf{Matching $r^{*}$} & \multicolumn{2}{c}{[0.72, 0.98]} \\
        \hline
        \textbf{Non-Matching $r^{**}$} & [-0.36, 0.39] & [-0.53, 0.68] \\
        \bottomrule
        \multicolumn{3}{l}{\footnotesize $^{*}$Between adversary-generated and target OoI traffic traces.} \\
        \multicolumn{3}{l}{\footnotesize $^{**}$Between adversary-generated and non-target traffic traces.}
    \end{tabular}
\end{table}

Table~\ref{tab:download-inducing} summarizes our attack evaluation results against the onion service in our institution.
We evaluate whether the adversary can identify the target OoI in the presence of non-target onion service traffic.
To obtain more evaluation samples without repeating the experiments too many times (due to ethical concerns), we conduct longer experiments and extract one-hour segments using a sliding window.
Specifically, we divide the paired adversary-generated and target OoI traffic traces into one-hour windows with a 12-minute overlap.
Each experiment therefore yields an average of 32.3 one-hour trace pairs.
For each window, the adversary compares the trace she generated with the corresponding target OoI trace and ten one-hour traces sampled from a non-target onion service.
The adversary then attempts to determine which of these 11 traces belongs to the target OoI.
Across all evaluation windows, the adversary correctly identifies the target OoI trace.
The correlation coefficients between the adversary-generated and target OoI traces range from 0.72 to 0.98, whereas those between the adversary-generated and non-target traces range from -0.36 to 0.39.

We repeat the attack against a target OoI hosted in the AWS North Virginia region (EC2 c5a.4xlarge instance), and Table~\ref{tab:download-inducing-aws} shows that the attack remains highly effective.
As before, we observe a 100\% success rate, with an even larger correlation-coefficient gap between matching and non-matching cases.
Note that the results should be interpreted in light of how the evaluation samples are constructed: the 11 one-hour traces are derived from one 25-hour trace, while the 97 and 29 evaluation windows come from three and one experimental runs, respectively.

\vspace{3pt}
\para{Evaluation II (large-scale, open).}
We incorporate real-world traffic traces to scale up our evaluation and test the attack in a more open setting.
We collect 1,203 real-world traces by operating our {\em own guard relay.}
At this relay, we record pseudonymized (within a single measurement session), one-minute-aggregated sampled packet counts for each host connecting to the relay over a 25-hour period.
We take special care for ethical considerations when collecting these traces under the guidance of Tor Research Safety Board~\cite{tor26research}.
Throughout the collection period, we do not inspect or store packet contents; only IP and port information from TCP/IP headers is accessed.
To further reduce risk, all IP addresses are always pseudonymized before being written to disk.

We use these large-scale real-world traces as additional false candidates.
For each evaluation window, the adversary compares the adversary-generated trace against 12,041 candidate traces: 1 trace from the target OoI, 10 traces from a non-target onion service, and 12,030 traces derived from the real-world Tor traffic.
Note that these real-world traces include not only onion services' traffic but also Tor clients' traffic and we do not differentiate them (just like real-world adversaries).
This way, we can evaluate our attack in the presence of real-world background onion services as well as clients.

Table~\ref{tab:download-inducing} shows that across all evaluation windows, the adversary successfully identifies the target OoI in our institution among all candidate traces.
Although the maximum correlation coefficient between the adversary-generated trace and the false candidates increases from 0.39 to 0.68, it remains below the minimum correlation coefficient between the adversary-generated trace and the target OoI trace (0.72).
Table~\ref{tab:download-inducing-aws} shows that the attack remains effective against the OoI hosted on AWS.
We again observe a 100\% success rate, with a larger gap in correlation coefficients between matching and non-matching cases.
Note that the 12,041 one-hour traces are derived from 1,204 25-hour traces.

These results show that download-inducing attacks are highly effective in deanonymizing onion services only with NetFlow traffic telemetry, even with real-world Tor traffic and over more than ten thousand candidate traces.
They also suggest that the high effectiveness of the download-inducing attack does not stem from artifacts of a specific experimental hosting environment.

\begin{table}[t]
    \setlength{\tabcolsep}{4pt}
    \renewcommand{\arraystretch}{1.20} 
    \caption{Additional Evaluation Results of Download-Inducing Attack on AWS}
    \label{tab:download-inducing-aws}
    \footnotesize 
    \centering
    \begin{tabular}{l|c|c}
        \toprule
        & \textbf{Small-Scale (Isolated)} & \textbf{Large-Scale (Open)} \\
        \midrule
        \makecell[l]{\textbf{Attack Setting}} & 1-of-11 & 1-of-12,041 \\
        \hline
        \textbf{Success Rate} & 100\% (29$/$29) & 100\% (29$/$29) \\
        \hline
        \textbf{Matching $r^{*}$} & \multicolumn{2}{c}{[0.86, 0.98]} \\
        \hline
        \textbf{Non-Matching $r^{**}$} & [-0.31, 0.35] & [-0.50, 0.65] \\
        \bottomrule
        \multicolumn{3}{l}{\footnotesize $^{*}$Between adversary-generated and target OoI traffic traces.} \\
        \multicolumn{3}{l}{\footnotesize $^{**}$Between adversary-generated and non-target traffic traces.}
    \end{tabular}
\end{table}

\vspace{3pt}
\para{Evaluation III (diverse, simulated).}
We evaluate how the duration and intensity of the attack affect its effectiveness.
For this evaluation, we design a simulation framework using data collected from our live experiments.
The evaluation incorporates 10,000 synthesized non-target onion service traffic traces (i.e., background-only traffic) together with one target OoI traffic trace.
For each configuration, we run the evaluation 100 times.

In the simulation, we assume that the adversary can induce traffic at the desired timing, consistent with our empirical observations presented earlier in this section.
The attack is launched in the default single-target mode and is controlled by two variables: duration and intensity.
Here, intensity is express as a multiplier of the average background traffic volume.
For example, an intensity of 1x means that the attack injects traffic equal to the mean background traffic rate at the desired timings.
We then compute the Pearson correlation coefficient between the expected traffic pattern, represented as a binary sequence indicating when traffic is induced, and each traffic trace.
We check whether the target OoI traffic trace achieves the maximum correlation coefficient.

We synthesize 10,000 non-target onion service traffic traces that model OoIs operating under normal conditions by randomly injecting background packets into base traffic traces measured without any attack or background traffic.
Specifically, the number of packets injected per minute is determined from the measured background-only traffic traces.
Because the traffic traces are ultimately converted into one-minute aggregated packet counts, they are largely insensitive to the fine-grained timing of individual packet injections.

In contrast, the target OoI traffic trace additionally includes injected attack traffic.
That is, non-target onion service traces consist of base traffic and background traffic, whereas the target OoI trace consists of base traffic, background traffic, and attack traffic.
The attack traffic is injected every two minutes, forming a square-wave traffic pattern with a 60-second active interval.

Both non-target onion service and OoI traffic traces are then sampled at a rate of 1:1,000.
We evaluate six different attack durations, 10 minutes, 1 hour, 10 hours, 1 day, 1 week, and 1 month (30 days), and five attack intensities, 0.1x, 0.5x, 1x, 2x, and 5x.

\begin{figure}[t!]
    \centering
    \includegraphics[width=\linewidth]{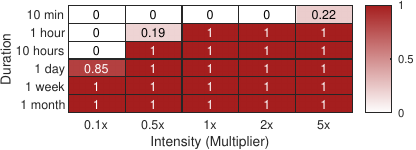}
    \caption{Deanonymization success rates for different attack durations and intensities.}
    \label{fig:simul-single}
\end{figure}

Figure~\ref{fig:simul-single} presents the deanonymization success rates across different attack durations and intensities.
We calculate the success rate as the fraction of trials in which the target OoI traffic trace is correctly identified among 10,001 candidate traces (i.e., 10,000 non-target onion service traces and one OoI trace).
Red indicates higher success rates while white indicates lower success rates.
Overall, the success rate increases with both attack duration and attack intensity.
The attack achieves a 100\% success rate with a one-week duration regardless of intensity, and also reaches a 100\% success rate with a one-hour duration when the intensity is at least 1x.
This trend is expected, as a longer or more intense attack directly produces a stronger traffic signal that is easier to distinguished from non-target background-only traffic.

\vspace{3pt}
\para{Evaluation IV (multi-target, simulated).}
We extend the attack from a single target to multiple concurrent targets using the simulation framework from the preceding evaluation.
So far, we have assumed that the adversary uses a single traffic pattern, namely the square-wave pattern, to identify one target OoI.
However, an adversary may wish to attack multiple OoIs simultaneously.
To support this setting, the adversary must assign a distinct traffic pattern to each target so that the induced signals can be separated during correlation analysis.
We therefore generate a codebook of balanced binary traffic patterns, where the first code corresponds to the default square-wave pattern and each subsequent code is selected to minimize its absolute correlation with the codes selected so far.
This reduces the likelihood that traffic induced for one target is mistakenly matched to another target's code.

This approach follows the principle of code-division multiplexing~\cite{gold67optimal,pursley77performance}, where multiple signals are transmitted over the shared channel by assigning each signal a distinct spreading code.
A receiver can recover a desired signal by correlating the received mixture with the corresponding code, provided that the codes have sufficiently low cross-correlation.
Analogously, in our setting, concurrent attacks share the same observation space, i.e., NetFlow traffic traces, and each target is assigned a distinct binary traffic code.
The detector separates targets by correlating each observed trace with each assigned code; therefore, low inter-code correlation is important for preventing one target's induced traffic pattern from being confused with another.

In this multi-target evaluation, we use three configurations: a one-hour duration with 0.5x intensity, a one-hour duration with 1x intensity, and a ten-hour duration with 0.5x intensity.
For each configuration, we vary the number of concurrent targets as 1, 2, 3, 4, 5, 10, 20, and 50.
For each target, we synthesize one target OoI traffic trace by injecting attack traffic according to its assigned code.
We then combine these target OoI traces with 10,000 non-target onion service traffic traces and compute the correlation between every candidate trace and every target code.
A target OoI is considered successfully deanonymized if its corresponding traffic trace achieves the highest correlation with the code assigned to that OoI.
Finally, we report the overall success rate across all targets and repetitions, along with the maximum absolute correlation among the generated codes as a diagnostic of codebook separability.

\begin{figure}[t!]
    \centering
    \includegraphics[width=\linewidth]{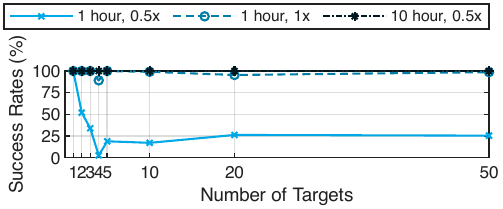}
    \caption{Deanonymization success rates under three multi-target settings.}
    \label{fig:simul-multi}
\end{figure}

Figure~\ref{fig:simul-multi} shows the deanonymization success rates of multi-target active attacks under three different configurations.
When the attack duration is one hour, the attack with 0.5x intensity begins to degrade as soon as the adversary targets multiple OoIs concurrently.
Increasing the intensity to 1x substantially improves the success rate, but the attack still does not consistently achieve a 100\% success rate once the number of concurrent targets exceeds three.
In contrast, with a ten-hour attack duration, even the 0.5x-intensity attack successfully deanonymizes all target OoIs across all evaluated target counts.
This again highlights the importance of attack duration: longer attacks provide more observations of the induced traffic pattern, making each target's code easier to distinguish from non-target onion service traffic and from the codes assigned to other concurrent targets.
Meanwhile, the maximum absolute correlation among the generated codes remains low: it is 0 for up to five targets, increases to 0.06, 0.13, and 0.2 for 10, 20, and 50 targets with one-hour codes, respectively, and remains even lower for ten-hour codes, reaching only 0.02, 0.03, and 0.05 for 10, 20, and 50 targets.

\vspace{3pt}
\para{Discussion.}
We now discuss three observations about the attack's effectiveness and practical implications:
\begin{packeditemize}
    \item {\em Impact of Tor's puzzle-based DoS defense.}
        Tor's puzzle-based DoS defense (see \S\ref{sec:puzzle-background}) has a negligible effect on download-inducing attacks.
        Repeating our main experiments with the puzzle-based DoS defense enabled yields the same results.
        This is expected because the download-inducing attack generates few introduction requests, while the defense is designed to mitigate floods of introduction requests.
        Consequently, the defense remains inactive throughout the experiments and does not affect the attack's effectiveness.
    \item {\em Impact of partial visibility.}
        As discussed in our threat model (\S\ref{sec:large-scale-arguments}), we consider a candidate IP address to be observable if the adversary can consistently observe packets in at least one direction between the IP address and one of its guard relays.
        To evaluate whether such partial visibility is sufficient in practice, we split each OoI traffic trace into ingress-only and egress-only traces and repeat the main evaluation.
        We find that partial visibility has little effect on attack effectiveness: egress-only traces achieve the same 100\% success rate as in the original evaluation, while ingress-only traces identify all but two targets.
        Although traffic volumes are highly asymmetric, the induced temporal pattern remains visible in both directions.
        A likely explanation is that sustained downloads generated corresponding reverse-direction traffic, including TCP acknowledgments and Tor protocol messages, which preserve the same alternating periods of activity and inactivity.
        We also observe that traffic is distributed evenly across relays, suggesting that the induced patterns remain observable even when only a subset of guard relays is monitored.
        Taken together, these results indicate that partial traffic visibility has a negligible impact on attack effectiveness.
    \item {\em Comparison with attacks against Tor clients~\cite{chakravarty14effectiveness}.}
        Our download-inducing attack has several advantages over the earlier attack against Tor clients~\cite{chakravarty14effectiveness}.
        First, targeting onion services requires fewer assumptions: the adversary does not need to lure a client to an adversary-controlled, decoy website or force it to download a large file, because the adversary herself can do so.
        Second, the candidate set is substantially smaller: whereas client deanonymization must distinguish a target from millions of users~\cite{dingledine14traffic}, our attack considers only tens of thousands of active, frequently accessed onion services.
        Third, our adversary can directly control the download process and thus precisely regulate the induced traffic pattern.
        These advantages help explain why our attack achieves highly effective identification in our large-scale evaluation, where the OoI is identified among more than ten thousand candidate traces.
\end{packeditemize}

\subsection{Introduction-Inducing Deanonymization}
\label{sec:intro-inducing}
A key limitation of the download-inducing attack is that a target OoI must support reliable downloads, allowing the adversary to induce the desired traffic pattern.
To eliminate this requirement, we propose an introduction traffic inducing (or introduction-inducing) attack against Tor onion services.
The core idea behind this new attack is to induce specific traffic patterns on the OoI by sending a number of introduction requests.
Because anyone can send introduction requests to an onion service with a known onion address and an available descriptor, the attack applies regardless of the application running on the OoI.

To launch such attacks, the adversary may either send a moderate number of introduction requests to carefully shape the desired traffic pattern (i.e., a low-rate attack) or send as many requests as possible to make the induced traffic clearly visible despite sampling (i.e., a high-rate attack).
When showing the feasibility of this attack, in both cases, we set up the adversary to repeatedly send batches of introduction requests interleaved with idle periods.

\vspace{3pt}
\para{Three attack phases.}
The introduction-inducing attack follows the same three phases as the download-inducing attack, as described in \S\ref{sec:attack-phases} and \S\ref{sec:download-inducing}.
The only difference is that, during Phase-\ding{203}, the adversary induces traffic by sending introduction requests to the target OoI rather than downloading a file from it.
For example, to induce a square-wave traffic pattern, the adversary alternates between sending introduction requests for 60 seconds and remaining idle for 60 seconds.
This is also a targeted attack, and thus each attack has a clear target definition.
Unless otherwise noted, we assume the default single-target attack mode in our evaluation; i.e., the adversary targets one OoI at a time.

\begin{table*}[tb]
    \setlength{\tabcolsep}{4pt}
    \renewcommand{\arraystretch}{1.20} 
    \caption{Evaluation Results of Introduction-Inducing Attack}
    \label{tab:intro-inducing}
    \footnotesize 
    \centering
    \begin{tabular}{l|c|c|c|c}
        \toprule
        & \multicolumn{2}{c|}{\bf Low-Rate (15~requests/min)} & \multicolumn{2}{c}{\bf High-Rate (1,383 requests/min)} \\
        \cline{2-5}
        & \textbf{Small-Scale (Isolated)} & \textbf{Large-Scale (Open)} & \textbf{Small-Scale (Isolated)} & \textbf{Large-Scale (Open)} \\
        \midrule
        \textbf{Attack Setting} & 1-of-11 & 1-of-12,041 & 1-of-11 & 1-of-12,041 \\
        \hline
        \textbf{Success Rate} & 100\% (98$/$98) & 98\% (96$/$98) & 24\% (23$/$96) & 13.5\% (13$/$96) \\
        \hline
        \textbf{Matching $r^{*}$} & \multicolumn{2}{c|}{[0.15, 0.93]} & \multicolumn{2}{c}{[-0.97, 1]} \\
        \hline
        \textbf{Non-Matching $r^{**}$} & [-0.28, 0.21] & [-0.57, 0.65] & [-0.38, 0.33] & [-0.78, 0.86] \\
        \bottomrule
        \multicolumn{5}{l}{\footnotesize $^{*}$Between adversary-generated and target OoI traffic traces.\quad$^{**}$Between adversary-generated and non-target traffic traces.}
    \end{tabular}
\end{table*}

\vspace{3pt}
\para{Evaluation I (small-scale, isolated).}
Similar to the previous section, we evaluate the effectiveness of introduction-inducing attacks on the live Tor network while initially considering only traffic traces generated within our experiment setup.
We deploy an onion service within our institution and on AWS, and conduct two experiments: one in which the adversary targets the OoI and another in which a separate onion service is not targeted, allowing us to collect non-target onion service traffic traces.
Our adversary runs a single Tor daemon on a single CPU core for the low-rate attack and multiple daemons across multiple CPU cores for the high-rate attack.
An attack script periodically sends introduction requests to the target OoI.

For each attack type (i.e., low-rate or high-rate), we conduct two 26-hour experiments under otherwise identical conditions.
In the first, the adversary alternates between sending introduction requests to the OoI for 60 seconds and remaining idle for 60 seconds.
In the second, the adversary remains inactive.
As before, we record one-minute aggregated packet counts near the onion service at a sampling rate of 1-out-of-1,000 in both experiments.
For the first experiment, we also record one-minute aggregated packet counts at the adversary's own machine (without sampling).
We again deploy five scripted clients to generate realistic background traffic and repeat each experiment twice, resulting in three runs for each.

Table~\ref{tab:intro-inducing} summarizes the attack evaluation results against a target OoI in our institution.
We follow the procedure described in \S\ref{sec:download-inducing}: paired adversary-generated and target OoI traces are divided into overlapping one-hour segments, and each adversary-generated trace is compared with the corresponding target OoI trace and ten non-target traces.
The adversary then attempts to identify the target OoI trace among the 11 traces.

First, under the low-rate attack, where the adversary sends an average of 15 introduction requests per minute, the adversary successfully identifies the target OoI among 11 candidate traces across all 98 evaluation windows.
The correlation coefficients between the adversary-generated traces and the target OoI traces range from 0.15 to 0.93, whereas those between the adversary-generated traces and the non-target onion service traces range from -0.28 to 0.21.
Notably, the separation between them is substantially smaller than that observed for the download-inducing attack (\S\ref{sec:download-inducing}), with some overlap, likely due to the indirect nature of introduction-request-based traffic induction.

Specifically, in the download-inducing attack, the adversary directly generates the desired traffic pattern, for example, by repeatedly downloading a large file.
In contrast, the introduction-inducing attack induces traffic indirectly by sending introduction requests.
The requests themselves contribute little traffic; instead, the observed pattern arises primarily from the target OoI's request-handling activity.

This finding is further supported by the high-rate attack evaluation.
When the adversary sends an average of 1,383 introduction requests per minute, the adversary fails to identify the target OoI among 11 candidate traces in 73 of the 96 evaluation windows.
Although some target OoI traces still correlate strongly with the adversary-generated traces, their correlation coefficients vary widely, from \mbox{-0.97} to 1.

Counterintuitively, increasing the request rate reduces attack effectiveness.
We attribute this degradation to overload at the target OoI caused by the excessive volume of introduction requests.
After receiving too many requests, the target OoI becomes unable to process new requests for extended periods of time (15.6 hours on average).
As a result, the traffic generated by the target OoI no longer follows the adversary-induced pattern.

We repeat the low-rate attack against a target OoI hosted in the AWS North Virginia region (EC2 c5a.4xlarge instance).
Table~\ref{tab:intro-inducing-aws} shows the same 100\% success rate, suggesting that the attack's high effectiveness is independent of the hosting environment.

\vspace{3pt}
\para{Evaluation II (large-scale, open).}
We now evaluate the introduction-inducing attack in the presence of real-world traffic traces.
In Table~\ref{tab:intro-inducing}, compared to the isolated setting, we observe a moderate reduction in attack effectiveness for the low-rate attack: the adversary identifies the target OoI among 12,041 candidate traces in 96 of the 98 evaluation windows.
Additionally, the maximum correlation coefficient between the adversary-generated traces and non-target onion service traces increases from 0.21 to 0.65.
The high-rate attack performs much worse: the adversary successfully identifies the target OoI in only 13 evaluation windows and fails in the remaining 83.
In practice, an adversary may repeat the low-rate attack over multiple runs, potentially with a longer attack duration, to identify OoIs with higher confidence.
In Table~\ref{tab:intro-inducing-aws}, where we repeat the low-rate attack against a target OoI hosted on AWS, we observe a 100\% success rate with a much larger correlation-coefficient gap between matching and non-matching cases.

These results demonstrate that introduction-inducing attacks can effectively deanonymize onion services with only NetFlow, even without relying on assumptions about the target OoIs.
They also show that low-rate attacks are substantially more effective than high-rate attacks due to the indirect nature of introduction-based traffic induction and that their high effectiveness is not specific to a particular hosting environment.

\begin{table}[ht]
    \setlength{\tabcolsep}{4pt}
    \renewcommand{\arraystretch}{1.20} 
    \caption{Additional Evaluation Results of Low-Rate Introduction-Inducing Attack on AWS}
    \label{tab:intro-inducing-aws}
    \footnotesize 
    \centering
    \begin{tabular}{l|c|c}
        \toprule
        & \textbf{Small-Scale (Isolated)} & \textbf{Large-Scale (Open)} \\
        \midrule
        \makecell[l]{\textbf{Attack Setting}} & 1-of-11 & 1-of-12,041 \\
        \hline
        \textbf{Success Rate} & 100\% (27$/$27) & 100\% (27$/$27) \\
        \hline
        \textbf{Matching $r^{*}$} & \multicolumn{2}{c}{[0.59, 0.92]} \\
        \hline
        \textbf{Non-Matching $r^{**}$} & [-0.26, 0.33] & [-0.61, 0.67] \\
        \bottomrule
        \multicolumn{3}{l}{\footnotesize $^{*}$Between adversary-generated and target OoI traffic traces.} \\
        \multicolumn{3}{l}{\footnotesize $^{**}$Between adversary-generated and non-target traffic traces.}
    \end{tabular}
\end{table}

\vspace{3pt}
\para{Evaluation III (diverse, simulated) and Evaluation IV (multi-target, simulated).}
The evaluations of attack duration, attack intensity, and concurrent targets presented for download-inducing attacks also apply to introduction-inducing attacks.
These evaluations use a simulation framework whose central assumption is that the adversary can induce traffic at the desired times.
Our empirical observations presented earlier confirm that this assumption also holds for introduction-inducing attacks.
Therefore, we refer readers to \S\ref{sec:download-inducing} for results under these diverse settings.

\vspace{3pt}
\para{Discussion.}
We discuss two observations about the attack's effectiveness and practical implications:
\begin{packeditemize}
    \item {\em Impact of Tor's puzzle-based DoS defense.}
        As with download-inducing attacks, Tor's puzzle-based DoS defense (see \S\ref{sec:puzzle-background}) has a negligible effect on low-rate introduction-inducing attacks.
        We repeat our main experiments with the defense enabled and obtain the same results.
        The low-rate introducing-inducing attack sends introduction requests at a sufficiently low rate, keeping the puzzle-based defense inactive throughout the experiments.
    \item {\em Impact of partial visibility.}
        Similar to download-inducing attacks, we evaluate partial visibility by splitting each OoI traffic trace into ingress-only and egress-only traces.
        We find that partial visibility has little effect on attack effectiveness, reducing the number of successful identifications by at most two.
        This result is expected because introduction requests trigger multiple bidirectional Tor protocol activities, producing similar traffic patterns in both directions.
        Additionally, we again observe that traffic is distributed evenly across relays, indicating that partial traffic visibility is largely sufficient for introduction-inducing attacks.
\end{packeditemize}
\section{{\em Passive} NetFlow-Based Deanonymization}
\label{sec:passive}
So far, all the discussed attack strategies have been some variants of {\em previously known}, {\em active} (because target onion services must receive attack traffic) attacks.
From now, we introduce a {\em novel}, {\em passive} (as absolutely no traffic sent to target onion service) deanonymization attack.
This passive nature, compared to active attacks, would be attractive to a cautious adversary seeking to avoid injecting its own traffic.

\para{Active vs. {\em passive} attacks.}
This new passive attack strategy also utilizes the same NetFlow-based coarse-grained traffic traces, just like the previously discussed active deanonymization attacks.
What differentiates our passive attack from active attacks is that the passive attack correlates the NetFlow traces of the candidate IPs with the {\em indirectly inferred traffic overheads} of the Onions of Interest (OoIs), {\em not} with the traffic sent by the adversary.
For the inference of traffic overheads at an OoI, we utilize a novel side channel that reveals the internal server load of OoIs.
We show that, by using this new side channel, an adversary can deanonymize OoIs without sending any traffic to them.
We begin by introducing this new side channel found in the current Tor onion services (\S\ref{sec:puzzle-background}) and providing an overview of the attack (\S\ref{sec:passive-overview}).
We then describe the attack procedure in detail (\S\ref{sec:semantic-matching}) and present our live evaluation results (\S\ref{sec:passive-eval}).

\subsection{A Novel Side Channel: Puzzle Difficulty of an Onion Service}
\label{sec:puzzle-background}
The core idea that makes our attack passive is the use of a new side channel of the {\em puzzle difficulty} of an OoI.
In August 2023, the Tor Project officially introduced client puzzles as a defense mechanism against DoS attacks targeting onion services~\cite{zoneff23introducing}, giving all onion services the option to enable this defense.
Specifically, the defense is designed to mitigate introduction-flooding attacks, in which an adversary attempts to overwhelm an onion service with a flood of introduction requests.
Tor's adoption of client puzzles was driven by the urgent need to mitigate such attacks, particularly in the absence of other effective and readily deployable countermeasures~\cite{lee25onions}.

The concept of client puzzles was first introduced by Juels and Brainard~\cite{juels99client} in 1999 as a specialized proof-of-work mechanism for defending against DoS attacks.
To access a server protected by client puzzles, clients must solve cryptographic challenges issued by the server.
While solving a small number of puzzles imposes only a modest computational burden on legitimate clients, it significantly increases the cost for attackers, who must solve a large number of puzzles to sustain an attack.

Unlike traditional client puzzle schemes, Tor's design incorporates several unique features to satisfy its strict anonymity requirements.
Notably, an onion service dynamically adjusts the difficulty of the puzzles it suggests to clients based on its current traffic load.
This difficulty level is then {\em broadcast to all} Tor clients via Tor directory servers, forming a global feedback loop.
Because it relies on the global feedback mechanism, the suggested difficulty is publicly available to anyone without requiring special privileges.
\footnote{In contrast, classic client puzzle mechanisms~\cite{juels99client} allow servers to assign client-specific puzzle difficulty levels.
Such an approach is infeasible for Tor onion services because, when a client sends an introduction request, there is no direct communication channel between the client and the onion service through which the service can assign a client-specific puzzle difficulty~\cite{cimpanu19tor,kadianakis20how}.
}

This creates {\em new side channel} information about the traffic overhead experienced by an onion service with the puzzle DoS defense activated.
Worse yet, this side channel information is made publicly available.
Every five minutes, the onion service evaluates its recent operational state and updates the suggested puzzle difficulty accordingly.
The onion service increases the suggested difficulty after receiving a large number of introduction requests, decreases it after receiving only a few requests, and leaves it unchanged if neither condition applies.
For details on the difficulty adjustment algorithm, we refer readers to the specification~\cite{tor26proof}, source code~\cite{tor26tor}, and a recent study~\cite{lee25onions}.

\subsection{Overview of Passive Deanonymization Attack}
\label{sec:passive-overview}
Our passive adversary performs Phase-\ding{202} reconnaissance in the same way as described in \S\ref{sec:attack-phases}.
Next, during Phase-\ding{203} NetFlow-based network measurement, the adversary collects the puzzle status of the OoIs along with network telemetry (i.e., one-minute aggregated packet counts).
Specifically, the adversary periodically queries directory servers to obtain the suggested puzzle difficulty levels of the OoIs.
Using these puzzle measurements (which we call {\em difficulty traces}), the adversary can infer the approximate traffic overhead of each OoI at the time of the query.
Note that collecting difficulty traces is inexpensive.
For example, even when monitoring 30,000 onion services, our measurements show that each query round, performed once per minute, requires roughly 2,160 CPU-seconds on a single core (base 2.4~GHz, max 3.5~GHz).
This is practical for a large ISP using parallel processing.
The storage cost is also modest, at about 1~GB/day; see \S\ref{sec:large-scale-discussion} for details.

Unlike the active attacks, this puzzle side-channel attack is {\em untargeted}; that is, instead of having a fixed, small set of OoIs, the attacks aim at deanonymizing any OoIs in a large set of OoIs chosen at the beginning.
This is because reading the difficulty values of many OoIs is low cost and scalable.

In Phase-\ding{204} target deanonymization, the adversary deanonymizes the OoIs by semantically matching (which is different from traditional traffic correlation; see \S\ref{sec:semantic-matching}) difficulty traces with traffic traces.
This phase can be performed after an observation period of several hours, days, or longer, and may be repeated any number of times while the adversary continues monitoring the network.
In practice, the adversary may continuously perform low-cost monitoring (i.e., Phase-\ding{203}) while executing the mapping process (i.e., Phase-\ding{204}) only occasionally or periodically (say every week or month).

\begin{figure}[!t]
    \centering
    \includegraphics[width=\linewidth]{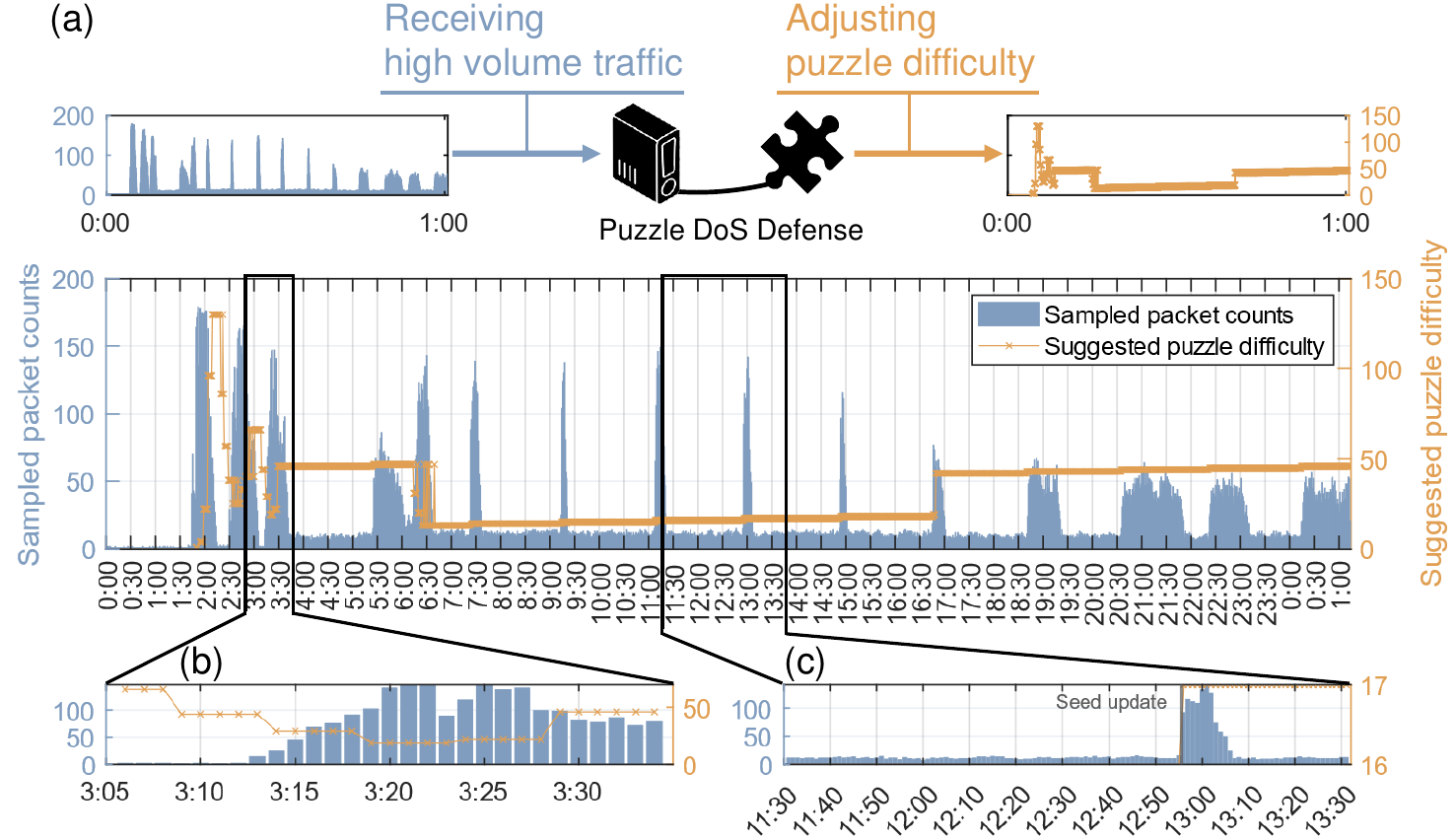}
    \caption{
    {\em Semantic matching} of difficulty and traffic traces for passive deanonymization of onion services.
    (a) {\em Real-world cause-effect relationship} between traffic volume and puzzle difficulty in the current Tor onion service.
    (b) A 30-minute zoom-in showing {\em invariants} of Tor's DoS defense system.
    (c) Another 30-minute zoom-in showing {\em artifacts} (e.g., sudden traffic spikes) found in the current DoS defense system.
    }
    \label{fig:examples}
\end{figure}

\subsection{{\em Semantic Matching} between Traffic and Puzzle Difficulty}
\label{sec:semantic-matching}
For the deanonymization of an OoI, the adversary needs a way to make {\em association} between two heterogeneous modality data.
A traffic trace is a measure that describes the traffic amount (in approximated packet counts) that a candidate IP address receives, while a difficulty trace records the puzzle's operating points determined by an OoI.
Figure~\ref{fig:examples} captures this distinction with a clear real-world example.
The traffic trace (blue) is the {\em input} to the puzzle-based DoS defense deployed at a server of an OoI (at the center), while the {\em resulting} difficulty trace (orange) shows how the server operates with difficulty adjustment.
The (rather expected) challenge of matching these two widely different traces is clear; that is, the two (blue and orange) traces do not seem to demonstrate the characteristics of typical correlated datasets.
For example, their ups and downs are different and even their longer term trends differ, rendering traditional statistical correlation techniques (e.g., lagged correlation) insufficient.
Our high-level intuition to this challenge is focus on the {\em casual relationship} between the two traces and resulting {\em invariants} found in the relationship, rather than their absolute values or localized changes.
If the traffic trace can explain the observed difficulty trace through the puzzle-based DoS defense mechanism, with the traffic trace acting as the {\em cause} and the difficulty trace as the {\em effect}, the adversary can infer that the onion address and candidate IP address likely correspond to the same OoI.

Understanding this casual relationship, let us summarize three key invariants in the relationship between the two traces denoted as {\bf (Inv-I)} to {\bf (Inv-III)}.

\vspace{2pt}
\para{(Inv-I): Curbing traffic load.}
The most fundamental invariant is the direct correlation between suggested puzzle difficulty and traffic load.
This follows the central principle of onion service DoS mitigation: when traffic volume increases (i.e., the service becomes busy), puzzle difficulty should rise accordingly, raising the computational cost of sending requests and thereby throttling incoming traffic.
For example, in Figure~\ref{fig:examples}(b), the suggested difficulty climbs from 22 to 46 at around 3:29, in response to a surge in sampled packet counts between 3:24 and 3:28.

\vspace{2pt}
\para{(Inv-II): Relaxing traffic load.}
Complementing {\bf (Inv-I)}, this invariant captures the inverse relationship: when traffic volume drops, puzzle difficulty should decrease to avoid imposing unnecessary burden on legitimate users.
In Figure~\ref{fig:examples}(b), the difficulty falls from 44 to 29 at 3:14, likely reflecting reduced traffic activity.

\vspace{2pt}
\para{(Inv-III): Non-zero difficulty.}
Yet another invariant stems from a subtle but highly reliable property of the suggested difficulty: a non-zero difficulty level necessarily implies the presence of significant traffic volume.
The fact that the puzzle-based DoS defense is activated in the first place indicates that the onion service has detected clear signs of attacks.
Unlike other invariants, this invariant is particularly useful to reject unlikely matches.
When non-zero difficulty is observed while the corresponding traffic volume remains consistently low, this mismatch signals an incorrect match.

In addition to invariants unaffected by the specifics of the DoS defense mechanism, we also identify two artifacts, {\bf (Art-I)} and {\bf (Art-II)}, that arise from particular design choices in Tor's DoS defense:

\vspace{2pt}
\para{(Art-I): Unchanging difficulty levels.}
This artifact arises from the cap on the number of concurrent connections being established and a threshold for broadcasting difficulty updates.
When this artifact is observed, the suggested difficulty remains unchanged for an extended period, even when the traffic volume is low.
This stagnation results from the minimum change threshold for broadcasting updates.
When the suggested difficulty is updated, any change less than 15\% from the original level is not broadcast but only applied internally.
As a result, a discrepancy arises between the internal and external views of the suggested difficulty. 
Because of the concurrent connection cap (explained below), the rate at which requests are processed can suddenly drop, which leads to prolonged periods of low traffic volume with unchanged difficulty levels.
In Figure~\ref{fig:examples}(c), this effect is visible as a sustained period, from 11:30 to 12:54, where sampled packet counts stay at or below 15, and the suggested difficulty remains constant at 16.

The reason for {the low traffic volume} is that the onion service limits the number of concurrent rendezvous circuits being established to 16 if a pending request solved a puzzle easier than the {internal} suggested difficulty.
Thus, even though the DoS adversary continues sending many requests that match the broadcast difficulty, the service treats them as easy (because the internally updated difficulty level is higher than the broadcast difficulty) and delays handling them.

\vspace{2pt}
\para{(Art-II): Traffic spikes following seed updates.}
The above artifact, {\bf (Art-I)}, is often accompanied by another artifact: a sudden spike in traffic following a prolonged period of low activity.
In addition to the suggested puzzle difficulty, an onion service also includes a {\em seed value} in its service descriptor: a randomly generated value used to prevent precomputation attacks.
This seed is rotated periodically, typically every 105 to 120 minutes.
When the seed is updated, the onion service ensures that all fields in the service descriptor, including the suggested difficulty, are up-to-date.
As part of this update, any previous discrepancy between the broadcast difficulty and the internally maintained value is resolved.
Once the updated internal suggested difficulty (previously unknown to external observers) is broadcast, the adversary can submit valid requests that solve puzzles at the correct difficulty level. 
This allows the onion service to resume establishing rendezvous circuits at its maximum rate, resulting in a sharp surge in traffic load.
This behavior is evident in Figure~\ref{fig:examples}(c), where the sampled packet counts rise sharply just after a seed update around 12:55.
At the same time, the suggested difficulty increases from 16 to 17.

Note that these artifacts are not invariants of the DoS defense mechanism and can thus be removed from Tor's puzzle defense implementation without degrading the defense effectiveness.
Yet, removing them has a negligible impact on our puzzle side-channel attack, as it can still deanonymize onion services with high accuracy by exploiting invariants; see \S\ref{sec:passive-mitigation}.

\vspace{2pt}
\para{Semantic matching procedure.}
For a given pair of difficulty and traffic traces, the adversary matches it by looking for semantic matching traits.
Specifically, the adversary quantifies the likelihood of a correct match based on the number of observed invariants and artifacts.
We define the {\em matching score} of a trace pair as the weighted sum of the occurrences of each invariant and artifact.
Formal definitions of the semantic matching process are provided in Appendix~\ref{app:formal}.

\vspace{2pt}
\para{Rejection mechanism.}
In practice, not all OoIs will have a perfect one-to-one correspondence with observable candidate IPs.
As a result, the adversary must be able to reject improbable matches.
We propose two rejection mechanisms.
First, the adversary discards a match if its matching score falls below a threshold $\tau_{reject}$.
This approach is based on the observation that incorrect matches tend to yield lower matching scores than correct ones.
Second, the adversary rejects a match if it violates {\bf (Inv-III)}, indicating a mismatch between observed traffic volume and puzzle difficulty.

\subsection{Evaluation}
\label{sec:passive-eval}
\para{Evaluation process.}
For each difficulty trace, we compute its semantic correlation with all available traffic traces and apply the rejection mechanisms.
If one or more pairs remain after rejection, we select the pair with the highest matching score as the final match for that difficulty trace.
This process yields at most one match for each difficulty trace.

\vspace{3pt}
\para{Trace collection.}
We collect OoI traffic and difficulty traces by deploying eight onion services (all located in our institution) on the live Tor network and sequentially launching DoS attacks against them.
We model introduction-flooding attacks in a simple but realistic manner: each attacker sends introduction requests to the onion service as fast as possible, consistent with real-world attacks observed against onion services~\cite{cimpanu19tor,kadianakis20how}.
Specifically, we initiate an introduction-flooding attack against each onion service after a random delay uniformly distributed between 0 and 12~hours.
Consequently, each onion service is targeted at a different time during the first 12 hours, after which all services are simultaneously under attack for the remainder of the experiment.

We record one-minute aggregated packet counts for each service with a 1:1,000 sampling rate.
At the same time, we query directory servers every minute to fetch the suggested puzzle difficulty and seed for each onion service.
To implement this periodic querying, we make a minor modification to Onionprobe~\cite{tor26onionprobe}, an open-source tool for monitoring onion service status.
Each experiment lasts around 25 hours and is repeated seven times, yielding 56 pairs of traffic and difficulty traces; we exclude one pair due to measurement errors, resulting in a final dataset of 55 trace pairs.

The onion services run on a server-grade CPU with unrestricted access to cores, whereas the DoS attackers operate on a separate server-grade CPU, with each attacker limited to twelve dedicated cores to minimize interference.
To incorporate background traffic, we first divide the background-only traffic trace collected in \S\ref{sec:traffic-inducing} into one-hour segments.
We then randomly select a segment and add it to each hour of every traffic trace.

We assume that the adversary has partial visibility into guard nodes and traffic directions as we consider a candidate IP address to be observable even if only a fraction of its traffic (e.g., packets {\em from} one of its guard nodes) transits the adversary network (see \S\ref{sec:large-scale-arguments}).
Accordingly, for each OoI, we randomly determine what fraction of its traffic (to/from each of its two guard nodes) is visible to the adversary.
This procedure is also used to mitigate potential overfitting: we repeat it separately for optimization and evaluation, ensuring that the two stages use distinct visibility assignments, thereby using distinct datasets.

Note that the number of collected trace pairs (55) is sufficiently large for evaluating our attack because this represents a scenario where 55 onion services are {\em simultaneously} subjected to DoS attacks.

\vspace{3pt}
\para{Evaluation scenarios.} We consider four evaluation scenarios, combining asynchronized and synchronized DoS attacks with perfect and limited adversarial knowledge.

Here, {\em asynchronized DoS attacks} refer to cases where attacks against different onion services begin independently, whereas {\em synchronized DoS attacks} refer to those where all attacks start simultaneously.
To simulate these two settings, we use full (i.e., 25-hour-long) traces collected throughout the entire experiment and truncated traces containing only data recorded after the first 12 hours, respectively.
The synchronized setting represents a worst-case scenario in which all onion services are attacked at the same time, making semantic matching significantly more difficult.
Therefore, the passive attack's real-world performance is expected to lie between the results observed in these two settings.

The {\em perfect-knowledge setting} posits that the adversary knows all relevant IP and onion addresses, such that every onion address can be mapped to a corresponding IP address.
In contrast, the {\em limited-knowledge setting} represents a more realistic scenario in which the adversary has incomplete knowledge and cannot associate some onion addresses with any candidate IP address.
Given that complete knowledge is improbable in practice, the real-world effectiveness of the attack is expected to fall between the results under these two settings.
To model limited knowledge, we randomly sample a subset of difficulty and traffic traces.

\vspace{3pt}
\para{Evaluation metrics.}
For the perfect-knowledge setting, we use the single metric, accuracy, defined as the ratio of correct matches to the total number of actual (ground-truth) matches.
For the limited-knowledge setting, we report two metrics: precision (the ratio of correct matches to the total number of predictions) and recall (the ratio of correct matches to the total number of actual matches).

\vspace{3pt}
\para{Optimization.}
To find the optimal parameters for our experiments, we apply Bayesian optimization~\cite{mockus78application}.
Specifically, we optimize five thresholds ($\tau_{InvI}$, $\tau_{InvII}$, $\tau_{ArtI}^{low}$, $\tau_{ArtI}^{high}$, and $\tau_{reject}$) and four weights ($w_{InvI}$, $w_{InvII}$, $w_{ArtI}$, and $w_{ArtII}$).
The objective function is defined as the sum of performance metrics across four evaluation scenarios: (1) asynchronized DoS attacks under the perfect-knowledge setting, (2) synchronized DoS attacks under the perfect-knowledge setting, (3) asynchronized DoS attacks under the limited-knowledge setting, and (4) synchronized DoS attacks under the limited-knowledge setting.
For the limited-knowledge settings, we fix the subset size to 40 during optimization.

For parameter bounds, we empirically constrain $\tau_{InvI}$ to lie between the 80th and 100th percentiles of the given traffic trace, and $\tau_{InvII}$ between the 0th and 20th percentiles.
More precisely, we define $\tau_{InvI}$ and $\tau_{InvII}$ as the $i$-th percentile of the traffic trace, where $i$ ranges from 80 to 100 and from 0 to 20, respectively.
{For $\tau_{ArtI}^{low}$ and $\tau_{ArtI}^{high}$, we let $i$ range from 0 to 100, while ensuring that $\tau_{ArtI}^{high}$ is greater than $\tau_{ArtI}^{low}$.} 
The weights are constrained to the interval $[0, 1]$, and $\tau_{reject}$ is bounded between 0 and 14.

We perform the optimization using 90 initial exploration steps, followed by 630 Bayesian optimization iterations.
This balances thorough exploration of the parameter space with efficient convergence toward optimal values.
As a result, we obtain the following parameter values:
\{$\tau_{InvI}$, $\tau_{InvII}$, $\tau_{ArtI}^{low}$, $\tau_{ArtI}^{high}$, $\tau_{reject}$, $w_{InvI}$, $w_{InvII}$, $w_{ArtI}$, $w_{ArtII}$\} = \{$95.08$, $3.49$, $22.13$, $94.65$, $1.22$, $0.15$, $0.05$, $0.37$, $0.57$\}.

Note that we use different datasets for optimization and final evaluation to avoid overfitting; particularly, we vary the fraction of each IP address's traffic that is visible to the adversary.

\begin{table}[t]
    \setlength{\tabcolsep}{2pt}
    \renewcommand{\arraystretch}{1.15} 
    \caption{Performance of Puzzle Side-Channel Attack}
    \label{tab:perf}
    \scriptsize
    \centering
    \begin{tabular}{cc cc cc}
        \toprule
        & \multirow{2}[2]{*}{\textbf{$|$Subset$|$}} & \multicolumn{2}{c}{\makecell{\textbf{Asynchronized DoS Attacks}\\(Realistic Scenarios)}} & \multicolumn{2}{c}{\makecell{\textbf{Synchronized DoS Attacks}\\(Contrived, Worst-Case Scenarios)}} \\
        \cmidrule(lr){3-4} \cmidrule(lr){5-6}
        & & \textbf{Precision} & \textbf{Recall} & \textbf{Precision} & \textbf{Recall} \\
        \midrule
        
        \multirow{6}{*}{\rotatebox{90}{Limited Knowledge}} 
        & 5 & 100\% (2$/$2) & 100\% (2$/$2) & 100\% (1$/$1) & 50\% (1$/$2) \\
        & 10 & 100\% (3$/$3) & 100\% (3$/$3) & 100\% (2$/$2) & 66.7\% (2$/$3) \\
        & 15 & 100\% (6$/$6) & 100\% (6$/$6) & 100\% (5$/$5) & 83.3\% (5$/$6) \\
        & 20 & 100\% (9$/$9) & 100\% (9$/$9) & 100\% (7$/$7) & 77.8\% (7$/$9) \\
        & 30 & 100\% (15$/$15) & 100\% (15$/$15) & 100\% (12$/$12) & 80\% (12$/$15) \\
        & 40 & 100\% (28$/$28) & 100\% (28$/$28) & 100\% (23$/$23) & 82.1\% (23$/$28) \\
        
        \midrule
        
        PK$^*$ & 55 & \multicolumn{2}{c}{100\% (55$/$55)$^{**}$} & \multicolumn{2}{c}{96.4\% (53$/$55)$^{**}$} \\
        
        \bottomrule
        
        \multicolumn{6}{l}{\footnotesize $^*$Perfect Knowledge.\quad $^{**}$Accuracy.}
    \end{tabular}
\end{table}

\vspace{3pt}
\para{Small-scale, isolated evaluation.}
We first evaluate the performance of the puzzle side-channel attack using only the OoI traffic and difficulty traces, both generated within our experiment setup.
Table~\ref{tab:perf} summarizes the evaluation results.
In the left two columns, under asynchronized DoS attacks, the attack achieves perfect precision and recall across all subset sizes in the limited-knowledge setting, as well as perfect accuracy in the perfect-knowledge setting.
In the right two columns, under synchronized DoS attacks, the attack maintains perfect precision across all subset sizes in the limited-knowledge setting, and achieves 96.4\% accuracy in the perfect-knowledge setting.

We highlight two key observations from these results.
First, the puzzle side-channel attack remains highly effective even in the challenging synchronized DoS attack setting.
This is particularly noteworthy because all synchronized DoS attacks follow the same realistic introduction-flooding strategy, demonstrating the worst-case performance of the attacks.
Second, the attack consistently achieves perfect precision across all scenarios.
This indicates that the adversary can reliably deanonymize OoIs with minimal risk of false positives, which is crucial for large-scale deployments.

We analyze the cases in which the attack fails to identify or misidentifies OoIs.
All observed failure cases occur under synchronized DoS attacks, where the adversary cannot utilize DoS onset timing.
In the limited-knowledge setting, the adversary fails to identify 1--5 targets because $\tau_{reject}$ is larger than the matching scores of the true positives.
In the perfect-knowledge setting, two targets are misidentified because a false candidate exhibits a similar traffic pattern and receives a higher matching score.

Additionally, we calculate relative score margins for all configurations.
In the limited-knowledge setting, the margin is defined as the normalized gap between the matching score and $\tau_{reject}$; in the perfect-knowledge setting, it is the normalized gap between the matching score and the highest non-matching score.
Across all subset sizes, the mean margins under limited knowledge are 39--43\% and 28--37\% for asynchronized and synchronized DoS attacks, respectively.
Under perfect knowledge, the corresponding mean margins are 60.1\% and 44.5\%.

\vspace{3pt}
\para{Large-scale, open evaluation.}
We scale up our evaluation in two ways: by incorporating real-world traffic traces collected at our guard relay (see \S\ref{sec:download-inducing}) and by further augmenting the resulting dataset with synthesized trace pairs.
For the latter, we synthesize 10,000 trace pairs that model OoIs operating normally (i.e., without DoS attacks).
This large set allows us to evaluate the performance of the attack at scale; i.e., when only a small fraction of OoIs are under DoS attacks while the vast majority are not.

First, we synthesize traffic traces by randomly injecting background packets into the non-DoS periods of each OoI traffic trace, similar to the simulation-based evaluations in \S\ref{sec:download-inducing}.
Next, we measure the suggested puzzle difficulty of the twelve real-world onion services (see \S\ref{sec:requirements}) and synthesize new difficulty traces by making variations in their timing and rotating difficulty values among services.
We heavily randomize the original difficulty traces from the real-world onion services to ensure that they are not attributable to any specific service.

The evaluation results are as follows.
First, under the perfect-knowledge setting, incorporating the real-world traffic traces does not change the puzzle side-channel attack's performance: it achieves {\em exactly the same} results as in the small-scale evaluation in Table~\ref{tab:perf}.
In the limited-knowledge setting, where no actual (ground-truth) matches exist due to sampling, the attack correctly rejects all matches, resulting in zero false positives.
These results show that unrelated real-world onion service traffic has no impact on the effectiveness of the attack, as such traces are filtered out by {\bf (Inv-III)} and the matching score threshold.
They also suggest that the attack works effectively even when the collected traffic traces contain Tor client traffic (which is included in the real-world traces), not only unrelated onion service traffic.
Second, when we further incorporate the synthesized trace pairs, the results remain unchanged from the previous open setting.
Since all normal-condition OoI traces are successfully filtered out, no false positives are introduced.

\begin{table}[t]
    \setlength{\tabcolsep}{2pt}
    \renewcommand{\arraystretch}{1.15} 
    \caption{Performance of Puzzle Side-Channel Attack on AWS}
    \label{tab:perf-aws}
    \scriptsize
    \centering
    \begin{tabular}{cc cc cc}
        \toprule
        & \multirow{2}[2]{*}{\textbf{$|$Subset$|$}} & \multicolumn{2}{c}{\makecell{\textbf{Asynchronized DoS Attacks}\\(Realistic Scenarios)}} & \multicolumn{2}{c}{\makecell{\textbf{Synchronized DoS Attacks}\\(Contrived, Worst-Case Scenarios)}} \\
        \cmidrule(lr){3-4} \cmidrule(lr){5-6}
        & & \textbf{Precision} & \textbf{Recall} & \textbf{Precision} & \textbf{Recall} \\
        \midrule
        
        \multirow{2}{*}{\rotatebox{90}{LK$^*$}} 
        & 2 & 100\% (1$/$1) & 100\% (1$/$1) & 100\% (1$/$1) & 100\% (1$/$1) \\
        & 4 & 100\% (2$/$2) & 100\% (2$/$2) & 100\% (1$/$1) & 50\% (1$/$2) \\
        
        \midrule

        PK$^{**}$ & 8 & \multicolumn{2}{c}{100\% (8/8)$^{***}$} & \multicolumn{2}{c}{100\% (8$/$8)$^{***}$} \\
        
        \bottomrule
        
        \multicolumn{6}{l}{\footnotesize $^{*}$Limited Knowledge.\quad\footnotesize $^{**}$Perfect Knowledge.\quad $^{***}$Accuracy.}
    \end{tabular}
\end{table}

\vspace{3pt}
\para{Different hosting environment evaluation.}
We also evaluate the impact of the hosting environment on the effectiveness of the puzzle side-channel attack.
Similar to \S\ref{sec:traffic-inducing}, we repeat the evaluation on the AWS North Virginia network using EC2 c5a.4xlarge instances.
Table~\ref{tab:perf-aws} summarizes the results.
The adversary achieves 100\% precision and accuracy regardless of subset size, showing a trend highly similar to the main results.
The mean relative score margins are even higher than those in the main results under the perfect-knowledge setting: 72.1\% and 64.9\% under asynchronized and synchronized DoS attacks, respectively.
In contrast, they are much smaller under the limited-knowledge setting: 24.6\% and 6.88\%.
This difference is likely due to the smaller sample size of this evaluation (8 vs. 55).
Based on these results, we conclude that, similar to the traffic-inducing attacks, the strong performance of the puzzle side-channel attack does not depend on the hosting environment.
\section{On Effective Mitigation Strategies}
\label{sec:mitigation}
Mitigating the proposed NetFlow-based deanonymization attacks, unfortunately, appears to be challenging.
Since removing access to large-scale, coarse-grained traffic telemetry (e.g., through legal or regulatory measures) at a global scale would be impractical, we focus on technical mitigations under the assumption that such large-scale telemetry is accessible to sufficiently capable adversaries.

\subsection{Traffic Obfuscation}
\label{sec:obfs}
Traffic obfuscation is perhaps a go-to defense against traffic analysis attacks in Tor~\cite{tian21universal}.
The Tor Project maintains an official obfuscator, obfs4~\cite{tor26lyrebird}, designed to circumvent censorship and mitigate traffic analysis attacks by padding packet sizes and introducing variability in packet inter-arrival times.
In the literature, Lopes et al.~\cite{lopes24flow} reported reductions of up to 20 percentage points in true positive rates, with similar performance degradation observed by Oh et al.~\cite{oh22deepcoffea} and Nasr et al.~\cite{nasr18deepcorr}.
Yet, existing traffic obfuscation techniques are entirely ineffective against the proposed NetFlow-based attacks because they modify traffic features at a fine-grained level (e.g., packet sizes, timing) that the NetFlow-based attacks do not exploit.

We confirm the ineffectiveness of traffic obfuscation through small-scale experiments on the live Tor network.
Specifically, we re-evaluate the effectiveness of the active download-inducing attack and the passive puzzle side-channel attack with obfs4 enabled.
As a pluggable transport designed for censorship circumvention~\cite{tor26circumvention}, obfs4 requires clients to connect through Tor bridges, specialized Tor relays whose information is not publicly listed~\cite{tor26bridge}.
We {\em opt not} to use existing bridges in our experiments, as they have limited resources compared to standard relays~\cite{tor26relay}, but instead set up our own bridges to avoid ethical concerns.
Additionally, to accommodate the constrained network conditions introduced by the use of our own bridges, we extend the experiment duration from 25 hours to around 37 hours for the passive attack evaluation.
This adjustment is informed by the observation that variability in suggested difficulty diminishes as attack traffic is funneled through a few Tor entry nodes (i.e., the 16 bridges).

\begin{figure}[!t]
    \centering
    \includegraphics[width=\linewidth]{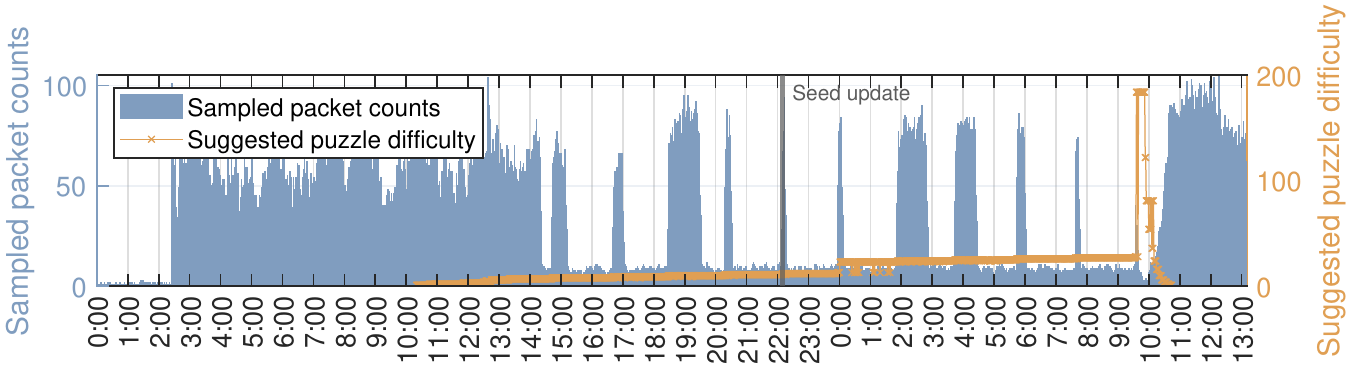}
    \caption{Sampled packet counts and suggested difficulty levels per minute, when obfs4 is deployed.}
    \label{fig:obfs}
\end{figure}

Our experiments confirm that, even with obfs4 deployed, the download-inducing attack successfully identifies the Onion of Interest (OoI) traffic trace among 12,041 candidate traces across all 31 windows, while the puzzle side-channel attack successfully identifies all (eight) onion-IP address pairs.
This result is further supported by Figure~\ref{fig:obfs}, which shows the sampled packet counts and suggested puzzle difficulty under obfs4.
One can observe that the sampled packet counts exhibit a clear semantic correlation with the suggested difficulty levels, as we observed in Figure~\ref{fig:examples} earlier. 
We therefore conclude that existing fine-grained traffic obfuscation is ineffective at mitigating NetFlow-based attacks.

\subsection{Detecting Induced Traffic Patterns}
\label{sec:detect}
At first glance, detecting the active attacks may seem straightforward because the target OoI has access to traffic at a finer granularity than the adversary.
Unlike the adversary, who observes only sampled and aggregated packet counts, the target OoI has access to unsampled raw traffic.\footnote{We empirically confirm that an OoI can detect known patterns with high accuracy, but omit the results for brevity.}

However, the adversary has a subtle but crucial advantage: she knows both the induced traffic pattern and the aggregation timing, whereas this information is {\em unknown} to the target OoI.
Specifically, the adversary is free to choose arbitrary traffic patterns (or codes) and attack timing.
Thus, while the adversary knows exactly which pattern to correlate and when to look for it, the OoI observes only its own traffic without knowing which pattern or timing to test.

This asymmetry creates a tradeoff between attack duration and detectability.
By inducing weaker signals over a longer period, the adversary can make the attack less distinguishable from ordinary background variation while still accumulating sufficient correlation evidence.
For example, a week-long or month-long low-intensity attack may be difficult for the OoI to distinguish from normal traffic fluctuations, even though an adversary correlating the trace with the correct code at the correct timing can still identify it reliably.
Consequently, detection alone may be insufficient to effectively mitigate active NetFlow-based deanonymization attacks, and additional measures are required.

\subsection{Disabling the Puzzle Side Channel}
\label{sec:passive-mitigation}
We explore several plausible strategies for significantly disabling (or even removing) the puzzle-based side channel, including artifact removal, time-domain adjustment, noisy difficulty adjustment, and coarser-grained difficulty adjustment.
Disabling this side channel is challenging because the relationship between traffic load and puzzle difficulty (i.e., the side channel itself) is fundamental to the DoS defense mechanism, and many strategies result in either weak mitigation effectiveness or weakened DoS protection.
Among these strategies, coarser-grained difficulty adjustment appears to balance DoS protection with side-channel mitigation.

\subsubsection{Removal of artifacts}
Since the artifacts observed earlier arise from specific design choices in Tor rather than from fundamental invariants, one seemingly reasonable mitigation strategy is to eliminate such artifacts.
Thus, in all subsequent experiments, we eliminate {\bf (Art-I)} and {\bf (Art-II)} by removing the constraints on difficulty broadcasting.
Our findings show that, even without these artifacts, the passive puzzle side-channel adversary can deanonymize onion services with high accuracy and precision by utilizing invariants.
We omit the experimental results that isolate artifact removal for brevity.

\subsubsection{Time-domain adjustment}
Given that artifact removal alone is insufficient, another natural direction is to reduce the responsiveness of difficulty adjustment by updating it less frequently; e.g., once per hour instead of every five minutes.
Yet, this approach has a clear drawback.
That is, it significantly weakens the effectiveness of puzzles as a DoS defense.
In particular, slower updates to the suggested difficulty prolong the system's reaction time to introduction-flooding, resulting in longer waiting times for clients and an increased window of opportunity for DoS attackers.
Understanding this clear limitation, we omit the empirical evaluation of this approach and instead focus on strategies that preserve the responsiveness of difficulty adjustment.

\subsubsection{Noisy difficulty adjustment}
\label{sec:noisy}
One strategy that preserves the responsiveness of difficulty adjustment is to obfuscate the suggested difficulty by injecting noise into it, a well-established countermeasure against side-channel attacks across various domains~\cite{osvik06cache,guneysu11generic,wang13role}.
In this approach, onion services use the existing difficulty calculation algorithm but modify the broadcast difficulty to obscure the true traffic conditions and confuse potential adversaries.
For example, an onion service samples a random value from a uniform distribution $U(0, 30)$ and adds it to the internally computed difficulty when broadcasting.
We deliberately choose this noise range for two reasons: first, restricting the noise to positive values preserves the effectiveness of puzzles as a DoS defense; second, capping the maximum noise at 30 limits the additional burden on clients, resulting in an average waiting time increase of 0.3 seconds, and at most 1.36 seconds even on resource-constrained mobile devices (see Appendix~\ref{app:theory} for details).
Using the same experiment setup described in \S\ref{sec:passive-eval}, we evaluate the effectiveness of this approach on the live network.

\begin{table}[t]
    \setlength{\tabcolsep}{2pt}
    \renewcommand{\arraystretch}{1.15} 
    \caption{Results of Noisy Difficulty Adjustment}
    \label{tab:noise}
    \scriptsize
    \centering
    \begin{tabular}{cc cc cc}
        \toprule
        & \multirow{2}[2]{*}{\textbf{$|$Subset$|$}} & \multicolumn{2}{c}{\makecell{\textbf{Asynchronized DoS Attacks}\\(Realistic Scenarios)}} & \multicolumn{2}{c}{\makecell{\textbf{Synchronized DoS Attacks}\\(Contrived, Worst-Case Scenarios)}} \\
        \cmidrule(lr){3-4} \cmidrule(lr){5-6}
        & & \textbf{Precision} & \textbf{Recall} & \textbf{Precision} & \textbf{Recall} \\
        \midrule
        
        \multirow{6}{*}{\rotatebox{90}{{Limited Knowledge}}} 
        & 5  & 100\% (2/2)   & 100\% (2/2)    & 100\% (1/1)     & 50\% (1/2)   \\
        & 10 & 100\% (3/3)   & 100\% (3/3)    & 66.7\% (2/3)    & 66.7\% (2/3) \\
        & 15 & 100\% (6/6)   & 100\% (6/6)    & 83.3\% (5/6)    & 83.3\% (5/6) \\
        & 20 & 100\% (9/9)   & 100\% (9/9)    & 72.7\% (8/11)   & 88.9\% (8/9) \\
        & 30 & 100\% (14/14) & 93.3\% (14/15) & 60\% (12/20)    & 80\% (12/15) \\
        & 40 & 100\% (26/26) & 92.9\% (26/28) & 63.6\% (21/33)  & 75\% (21/28) \\
        
        \midrule
        
        PK$^*$ & 55 & \multicolumn{2}{c}{94.6\% (52/55)$^{**}$} & \multicolumn{2}{c}{76.4\% (42/55)$^{**}$} \\
        
        \bottomrule
        
        \multicolumn{6}{l}{\footnotesize $^*$Perfect Knowledge.\quad $^{**}$Accuracy.}
    \end{tabular}
\end{table}

Unfortunately, however, our results show that noise injection is largely ineffective.
In fact, the added noise introduces a new artifact enabling accurate trace matching, as we discuss next.
Table~\ref{tab:noise} presents the performance of the puzzle side-channel attack with noise injection.
While the attack's overall effectiveness declines across settings, its precision (in the limited-knowledge setting) and accuracy (in the perfect-knowledge setting) remain (near-)perfect under asynchronized DoS attacks.

We attribute the limited effectiveness of the noise-based approach to two main factors.
First, noise affects not only puzzle side-channel adversaries but also DoS attackers.
Because DoS attackers adapt their behavior based on the broadcast difficulty, the fundamental correlation between traffic load and suggested difficulty persists.
Second, this mitigation strategy introduces a new artifact, {\bf (Art-III)}, caused by overly frequent difficulty updates.
Under normal conditions, an onion service uploads its descriptor to multiple directory servers to ensure availability.
However, if descriptor updates occur too frequently, some directory servers may receive the updated descriptors while others retain outdated ones.
This inconsistency leads to divergent difficulty values across directories, which in turn can reduce traffic volume via the same mechanism that caused {\bf (Art-I)}.
Although such inconsistencies are rare under standard operation, they become increasingly common when the noise-based strategy triggers frequent descriptor updates; see Appendix~\ref{app:formal} for a formal definition.

The re-optimized parameters for noisy difficulty adjustment are:
\{$\tau_{InvI}$, $\tau_{InvII}$, $\tau_{ArtIII}$, $\tau_{reject}$, $w_{InvI}$, $w_{InvII}$, $w_{ArtIII}$\} = \{$72.79$, $5.84$, $17.33$, $16.23$, $0.79$, $0.20$, $0.51$\}.
We use updated parameter bounds for $\tau_{InvI}$ and $\tau_{ArtIII}$: (70, 90) and (10, 30), respectively.

\subsubsection{Coarser-grained difficulty adjustment}
\label{sec:coarser}
Like the noisy difficulty adjustment strategy, this approach retains the existing difficulty calculation algorithm but modifies the broadcasted difficulty to be coarser-grained.
Importantly, it also preserves the effectiveness of the puzzle mechanism as a DoS defense.
In this strategy, the onion service always rounds up its internally computed suggested difficulty to the nearest multiple of fifteen before broadcasting it.\footnote{The value fifteen is chosen empirically.}
Since the maximum difference between the broadcast and internal difficulty is {\em capped at} 14, the added delay for clients on resource-constrained mobile devices is limited to an average of 0.14 seconds and 0.63 seconds in the worst case (see Appendix~\ref{app:theory} for the estimation process).

\begin{table}[t]
    \setlength{\tabcolsep}{2pt}
    \renewcommand{\arraystretch}{1.15} 
    \caption{Results of Coarser-Grained Difficulty Adjustment}
    \label{tab:mult}
    \scriptsize
    \centering
    \begin{tabular}{cc cc cc}
        \toprule
        & \multirow{2}[2]{*}{\textbf{$|$Subset$|$}} & \multicolumn{2}{c}{\makecell{\textbf{Asynchronized DoS Attacks}\\(Realistic Scenarios)}} & \multicolumn{2}{c}{\makecell{\textbf{Synchronized DoS Attacks}\\(Contrived, Worst-Case Scenarios)}} \\
        \cmidrule(lr){3-4} \cmidrule(lr){5-6}
        & & \textbf{Precision} & \textbf{Recall} & \textbf{Precision} & \textbf{Recall} \\
        \midrule
        
        \multirow{6}{*}{\rotatebox{90}{{Limited Knowledge}}} 
        & 5 & 20\% (1$/$5) & 100\% (1$/$1) & 20\% (1$/$5) & 100\% (1$/$1) \\
        & 10 & 37.5\% (3$/$8) & 100\% (3$/$3) & 20\% (2$/$10) & 66.7\% (2$/$3) \\
        & 15 & 40\% (6$/$15) & 100\% (6$/$6) & 33.3\% (5$/$15) & 83.3\% (5$/$6) \\
        & 20 & 42.1\% (8$/$19) & 88.9\% (8$/$9) & 35\% (7$/$20) & 77.8\% (7$/$9) \\
        & 30 & 51.7\% (15$/$29) & 93.8\% (15$/$16) & 42.9\% (12$/$28) & 75\% (12$/$16) \\
        & 40 & 70\% (28$/$40) & 96.6\% (28$/$29) & 46.2\% (18$/$39) & 62.1\% (18$/$29) \\
        
        \midrule
        
        PK$^*$ & 55 & \multicolumn{2}{c}{92.6\% (50$/$54)$^{**}$} & \multicolumn{2}{c}{57.4\% (31$/$54)$^{**}$} \\
        
        \bottomrule
        
        \multicolumn{6}{l}{\footnotesize $^*$Perfect Knowledge.\quad $^{**}$Accuracy.}
    \end{tabular}
\end{table}

Using the same experiment settings described in \S\ref{sec:passive-eval}, we evaluate the effectiveness of this coarser-grained difficulty adjustment strategy on the live network.
Our results suggest that this approach has the potential to serve as an effective defense against the puzzle side-channel attack.
Table~\ref{tab:mult} summarizes the evaluation results.
Notably, precision drops sharply across all settings, while recall degrades slightly under asynchronized DoS attacks and more substantially under synchronized attacks.
The low precision indicates that this strategy can effectively mitigate the practical deanonymization attack by introducing a large number of false positives for the adversary.
Meanwhile, the high accuracy (near 93\%) observed for asynchronized DoS attacks under the perfect-knowledge setting suggests that the onset timing of DoS attacks remains a critical signal for accurate trace matching in the puzzle side-channel attack.
This emphasizes the need to obscure such timing information; one possible direction is for onion services to periodically broadcast non-zero difficulty values even in the absence of DoS attacks, at randomized intervals.

The re-optimized parameters for coarser-grained difficulty adjustment are:
\{$\tau_{InvI}$, $\tau_{InvII}$, $\tau_{reject}$, $w_{InvI}$, $w_{InvII}$\} = \{$77.07$, $11.67$, $1.09$, $0.97$, $0.99$\}.
We use an updated parameter bound for $\tau_{InvI}$: (70, 90).

\subsection{Rotating IP Addresses}
\label{sec:rotation}
One alternative, orthogonal to all previous mitigations, is to restrict the information that an adversary can obtain from such monitoring via moving target defense~\cite{jajodia11moving,sengupta20survey}.
Specifically, onion services could rotate their IP addresses periodically and in response to detected suspicious traffic patterns, preventing adversaries from collecting sufficient telemetry.
This periodic-and-responsive rotation could potentially mitigate both active and passive attacks.
For short, high-intensity active attacks, IP rotation upon detection would limit attack effectiveness.
For long, low-intensity active attacks and passive attacks, periodic rotation would limit attack effectiveness.

The specific mechanism of IP rotation may vary depending on the onion service's hosting environment.
Onion services hosted on cloud or hosting networks may be able to rotate their IP addresses relatively easily, while those hosted on local networks may have difficulty continuously obtaining new addresses.
In such cases, coordination with the local ISP or migration to a cloud or hosting provider might be beneficial.

The required rotation granularity depends on the adversary's monitoring strategies and scope.
Changing only a few bits (or even a single bit) of the IP address may be sufficient against an adversary that does not account for IP rotation.
A more sophisticated adversary, however, may counter such naive rotation by grouping addresses by subnet.
Rotating across substantially different address ranges would be required against such adversaries.

\subsection{Overall Suggestions}
Our findings suggest that (i) traditional traffic obfuscation would be largely ineffective; (ii) conspicuous active attacks can often be detected, but low-amplitude induction remains difficult to distinguish from natural traffic variation; (iii) disabling the puzzle side channel introduces a slight availability-anonymity trade-off; and (iv) IP rotation would help mitigate both active and passive NetFlow-based deanonymization attacks although it may require non-trivial operational effort.
Therefore, active-attack detection, disabling the puzzle side channel, and IP rotation would have to be considered together, as no single defense appears to fully mitigate NetFlow-based deanonymization attacks.
We have shared our implementations and experimental results with the Tor security team to support their evaluation of potential mitigations.
\section{Discussion}
\label{sec:discussion}

\subsection{Limitations}
Although our experiments are conducted on the live Tor network, they do not capture the full diversity of onion service deployments and attack conditions in the wild.
For ethical reasons, our evaluations rely primarily on a controlled experiment setup with onion services and adversarial infrastructure under our control, and we derive NetFlow-like records from packet traces rather than collecting actual NetFlow exports from a hardware router.
Real onion services may differ in resources, hosting environments, locations, and application behavior; router-specific NetFlow implementations may also introduce additional artifacts not captured in our setup.
We incorporate real-world Tor traffic traces ethically collected at our guard relay to improve realism, but these traces still cannot fully represent all deployment and traffic characteristics of real onion services.
Therefore, our results should be interpreted as demonstrating the feasibility and practicality of NetFlow-based onion service deanonymization by a network adversary, rather than as a precise measurement of deanonymization effectiveness across the entire Tor network.

\begin{table*}[t]
    \setlength{\tabcolsep}{4pt}
    \renewcommand{\arraystretch}{1.20} 
    \caption{Effectiveness of the Studied NetFlow-Based Attacks under Varying Sampling Rates}
    \label{tab:sampling-rates}
    \footnotesize 
    \centering
    \begin{tabular}{lccccc}
        \toprule
        & \textbf{1:500} & \textbf{1:1000 (Main)} & \textbf{1:1500} & \textbf{1:2000} & \textbf{1:5000} \\
        \midrule
        \textbf{Download-Inducing$^{*}$} & 100\% (97$/$97) & 100\% (97$/$97) & 100\% (97$/$97) & 100\% (97$/$97) & 95.9\% (93$/$97) \\
        \textbf{Low-Rate Introduction-Inducing$^{*}$} & 98\% (96$/$98) & 98\% (96$/$98) & 96.9\% (95$/$98) & 94.9\% (93$/$98) & 66.3\% (65$/$98) \\
        \textbf{Puzzle Side-Channel$^{**}$} & 98.2\% (54$/$55) & 96.4\% (53$/$55) & 92.7\% (51$/$55) & 83.6\% (46$/$55) & 61.8\% (34$/$55) \\
        \bottomrule
        \multicolumn{6}{l}{\footnotesize $^{*}$Large-scale, open evaluation.\quad$^{**}$Accuracy under perfect knowledge and synchronized DoS attacks.} \\
    \end{tabular}
\end{table*}

\subsection{Different Attack Configurations}
In this study, we focus on 1:1000 sampling, one-minute aggregated packet counts, and Pearson correlation to demonstrate a low-cost, low-complexity attack scenario under common NetFlow settings.
This configuration emphasizes the practicality of the attacks: (i) a 1:1000 sampling rate is commonly suggested by the operator community~\cite{popa23what,petryschuk24netflow,cloudflare25recommended,fastnetmon25fastnetmon}, (ii) one-minute aggregation is a default or recommended reporting interval~\cite{petryschuk24netflow,juniper24netflow}, (iii) packet count is the simplest feature representation, and (iv) Pearson correlation is one of the simplest correlation techniques.
In this section, we discuss other, less common attack configurations, especially those using different sampling rates.

In general, higher sampling rates (i.e., capturing more packets) provide more frequent observations and thereby improve attack effectiveness.
We empirically confirm this trend by varying the sampling rate among 1:500, 1:1000 (our main configuration), 1:1500, 1:2000, and 1:5000.
Table~\ref{tab:sampling-rates} summarizes the performance of the three attacks, download-inducing, low-rate introduction-inducing, and puzzle side-channel attacks, under these sampling rates.
For the sake of brevity, for the active attacks, we report only deanonymization success rates from the large-scale, open evaluation; for the passive attack, we report only accuracy under perfect knowledge and synchronized DoS attacks.
All three attacks exhibit a clear trend of decreasing performance as the sampling rate decreases, while the download-inducing attack remains highly accurate even at a sampling rate of 1:5000.
These results suggests that the sampling rate can be chosen based on the adversary's available resources and desired attack accuracy.
Future work may explore other configurations, including different aggregation intervals, feature representation, and correlation techniques.

\subsection{Timing Imperfections}
Timing imperfections may arise from multiple sources over the course of the attacks.
In this section, we discuss how such imperfections are accounted for.
First, although differences in device clocks may disrupt traffic monitoring by distorting traffic aggregations, their impact is limited because the attacks require only minute-level granularity and the adversary is assumed to control centralized network monitoring infrastructure; see \S\ref{sec:threat-model}.
Second, our semantic matching rules account for variable delays between changes in traffic load and the publication of corresponding difficulty levels, as well as the five-minute puzzle update cycle; see Appendix~\ref{app:formal}.
Last, propagation delays introduced by directory servers are inherently included in our evaluations, as we fetch descriptors from many real Tor directory servers on the live network; see \S\ref{sec:passive-eval}.

\subsection{Opportunistic Nature of the Passive Attack}
\label{sec:requirements}
The puzzle side-channel attack is passive and opportunistic: it deanonymizes OoIs when they are under DoS attacks launched by any party, not necessarily by the adversary.
Although this may require waiting until suitable DoS conditions arise, we argue that this does not significantly limit the practicality of the attack for the following reasons.

First, DoS attacks against onion services are prevalent.
Many incidents have been reported over the years~\cite{brewster15tor,cimpanu19tor,kadianakis20how,bagueros23tor}, and Tor's adoption of client puzzles was indeed prompted by an increasing number of DoS attacks on onion services.
We conducted a preliminary measurement study to confirm the prevalence of DoS attacks on the current Tor network.
From May 30 to November 2, 2025, we observed real-world puzzle difficulty levels for twelve well-known onion services that had enabled the puzzle defense, either ranked in the top five of their category on the Tor Link search engine~\cite{tor26link} or identified manually.
Among these services, five exhibited non-zero difficulty at least once, and two showed clear variation in difficulty: one ranged from 1,000 to 12,000 with a standard deviation of 2,206, and the other ranged from 0 to 400 with a standard deviation of 146.

Second, the puzzle side-channel attack can be used as an always-on surveillance tool for the adversary that controls network infrastructure (e.g., ISPs, nation-state adversaries) and conducts long-term deanonymization campaigns, in which case the opportunistic nature is less of a concern.

\subsection{Considerations for Enhanced Attacks}
\para{Bridges.}
Potential limitation of the proposed deanonymization attacks is the existence of non-public relays, known as Tor bridges~\cite{tor26bridge}.
During the reconnaissance phase, some Tor-speaking IP addresses may be overlooked if they communicate exclusively through bridges, whose information is kept secret{; or, some bridges may be misclassified as non-relay Tor-speaking addresses}.
At first glance, this may seem to pose a significant challenge.
However, historical efforts by authoritarian states to uncover bridges~\cite{ling12extensive,wilde12knock,dunna18analyzing} suggest that even these addresses could eventually be discovered and incorporated into the adversary's monitoring list.
Additionally, for onion services, relying solely on bridges may result in degraded service quality due to bridges' limited bandwidth and computational resources~\cite{tor26relay}.

\vspace{3pt}
\para{Repetitive attacks for higher accuracy.}
The adversary can repeat the attacks to improve accuracy.
Our current evaluation infers a linkage between an onion address and an IP address from a single observed trace match.
In practice, the adversary could monitor and match traces repeatedly, increasing confidence in onion-IP associations that recur across multiple observations~\cite{wails24precisely}.

\subsection{Large-Scale Feasibility of Proposed Attacks}
\label{sec:large-scale-discussion}
In this section, we extend the discussion of the large-scale feasibility of the proposed attacks from \S\ref{sec:large-scale-arguments} to puzzle data collection and actual target deanonymization (i.e., Phase-\ding{204}).
Again, to assess the resource requirements of attacks at the largest scale, we assume that an adversary aims to deanonymize as many major onion services as possible.
Consequently, we assume that the adversary monitors roughly 30,000 candidate IPs and 30,000 onion addresses.

First, the cost of collecting difficulty traces is modest.
A passive adversary retrieves the suggested difficulty (and optionally seeds) for each onion service from its public service descriptors, which can be fetched via queries to Tor directory servers.
In case of targeting 30,000 onion services, the adversary can distribute the queries across all (roughly 5,000) active Tor directory servers~\cite{tor26onionoo} to avoid overloading any single server.
This requires initially building 5,000 Tor circuits and issuing 30,000 queries per minute.
Based on our measurements of RAM usage and CPU time for circuit creation and descriptor retrieval, we estimate that building 5,000 circuits requires around 70~MB of RAM and 300 seconds of CPU time, while fetching 30,000 descriptors requires about 480~MB of RAM and 2,160 seconds of CPU time.
Also, storing difficulty traces for 30,000 onion services would require 780~KB/minute, 1~GB/day, or 34~GB/month.
Therefore, both computational and storage requirements are well within the capabilities of a large ISP.

Second, the time required for actual target deanonymization is also modest.
For active attacks, the adversary needs to compute Pearson correlation coefficients for approximately 30,000 trace pairs.
Based on our measurements, and assuming a one-hour trace duration, this computation takes less than a minute even on a single core of a server-grade CPU (base frequency 2.4~GHz, max 3.5~GHz).
For passive attacks, the adversary needs to semantically match 30,000 traffic traces with 30,000 difficulty traces, resulting in 900 million trace pairs to evaluate.
The adversary would wish to complete a full semantic matching cycle (i.e., all 900 million evaluations) within 25 hours for continuous monitoring and matching.
Based on our measurements, and assuming a 25-hour trace duration, this computation takes about 620 minutes on eight CPU cores.
Here, the length of the traces has minimal impact on feasibility, as the trace duration aligns with the time available for processing.
Therefore, we conclude that target deanonymization is also readily feasible within a practical amount of time at ISP scale.

\subsection{Applicability to Other Anonymity Systems}
This work focuses on deanonymizing Tor onion services, given Tor's widespread use and official adoption of client puzzles.
Yet, our findings likely extend to other anonymity networks that provide receiver anonymity, such as the Invisible Internet Project (I2P)~\cite{i2p26the}.
The proposed attacks do not exploit Tor-specific design flaws or implementation bugs; rather, it leverages fundamental side channels, the traffic volume (for active attacks) and the correlation between traffic volume and suggested puzzle difficulty (for passive attacks).
Therefore, the active attacks would be already effective against anonymous services in other anonymity networks, and passive attacks would become effective once those networks introduce client puzzles.
Our findings suggest that anonymity networks should carefully evaluate the privacy implications of NetFlow-based deanonymization attacks.
\section{Conclusion}
\label{sec:conclusion}
We revisit and extend the largely neglected threat model of NetFlow-only deanonymization, revealing its renewed threats to Tor onion services.
We demonstrate that active deanonymization of onion services using sampled NetFlow is not only feasible but also practical, and that even NetFlow-based passive deanonymization is possible.
Our findings suggest that anonymity networks should carefully re-evaluate the threat that this long-standing adversarial capability poses to anonymous service providers and seek comprehensive mitigation strategies.
\section*{Ethical Considerations}
\label{app:ethics}
First, the computational and memory overhead imposed on Tor relays is minimal, as our experiments target onion services under our direct control, with Tor relays serving only as intermediate nodes along the communication path.
Also, the load is distributed across 20--160 introduction points, which further mitigates the impact.
The peak connection rate per relay in our experiments never exceeds 16 simultaneous connections, which is only 1.1\% of the minimum required capacity (i.e., 7,000) for Tor relays~\cite{tor26relay}.

Second, the total traffic generated by our experiments remains low.
A single day-long experiment may generate up to 16,032,000 requests, corresponding to a maximum total volume of approximately 180~GB, which is distributed across 160 introduction points.
Assuming the experiment is repeated fifteen times in a month, the total traffic per relay would amount to 0.8\% of the recommended monthly capacity (2~TB)~\cite{tor26relay}.

Last, the potential impact on relays acting as Tor directory servers is minimal.
Each heavy experiment performs periodic descriptor fetching at one-minute intervals for eight onion services, resulting in at most 8 simultaneous connections (0.11\% of the minimum required relay capacity~\cite{tor26relay}) and approximately 12,000 fetch requests, generating less than 135~MB of traffic (0.007\% of the recommended monthly capacity~\cite{tor26relay}).

In addition to the resources consumed by the relays, we take several practical precautions to minimize any potential negative impact on the Tor network during our experiments.
(i) We deploy our own onion services and configure them to use as many introduction points as possible (20 in our experiments) to effectively distribute the generated traffic.
(ii) When assessing the effectiveness of traffic obfuscation as a mitigation strategy (i.e., in \S\ref{sec:obfs}), we operate our own Tor bridges.
This is particularly important because Tor bridges generally have significantly fewer resources than typical middle relays~\cite{tor26relay}.
Although this approach extends the duration of the experiments due to the traffic convergence issue, we consider it a necessary trade-off to avoid burdening legitimate Tor bridges.
(iii) We limit the number of intensive experiments to fifteen per month, consistent with the assumption made in our analytic evaluations.

\para{Traffic collection.}
Before operating a public relay and collecting one-minute-aggregated, sampled, and pseudonymized traffic traces, we consulted the Tor Research Safety Board~\cite{tor26research}.
We shared our detailed collection plan, including ethical considerations, and the board reviewed our procedures.
Based on their feedback, we revised our plan to address all concerns raised.
After four rounds of revisions, we obtained approval from the board and proceeded with our experiments, strictly following the approved plan.

The following presents our hypothetical worst-case risk analysis.
In the worst-case scenario, where an adversary gains access to our one-minute-aggregated packet counts (e.g., through a data breach), our data collection would pose no practical harm to the Tor network or its users.
For an adversary to correlate the obtained packet counts, they would need access to puzzle difficulty data from onion services at a fine temporal granularity (at least once per minute).
To our knowledge, no mechanism exists for retrospectively collecting such puzzle difficulty data from arbitrary onion services, making it challenging for an adversary to obtain the required data after the fact.

Even if an adversary were somehow able to collect the necessary puzzle difficulty data, the resulting output would consist only of {\em pseudonymized} onion-IP address linkages, which cannot be deanonymized without the pseudonymization key.
This key is used solely during raw data collection, sampling, and aggregation, and is discarded immediately after the collection period ends.
No copies of the key are stored on disk or transmitted elsewhere.
Although an adversary might infiltrate our storage server to obtain the collected traffic traces, it is highly unlikely that they could also obtain the key, because the window of opportunity is extremely narrow (existing only during the active collection period).
Accordingly, we assess the risk of IP deanonymization as minimal.

We acknowledge that the risk is non-zero, though.
A highly sophisticated, targeted adversary could potentially spear-phish our relay operator during the collection period to obtain the key, thereby acquiring both the data and the key needed to deanonymize IP addresses.
However, we consider this scenario extremely unlikely.
An adversary capable of compromising the relay operator during the experiment would likely already be able to access raw, unsampled, and non-anonymized packet data directly, bypassing the need for the key entirely.
Alternatively, such a resourceful adversary could operate their own entry guard relays to collect detailed raw data independently.

Building on this worst-case risk analysis, we also present a harm-benefit analysis.
As explained above, we believe that the risks associated with this data collection are minimal.
Even in a rational worst-case scenario, an adversary could obtain only {\em pseudonymized} onion-IP address linkages, without any practical means to deanonymize the addresses.
In contrast, the benefits are substantial.
This dataset enables a rigorous evaluation of the practicality of NetFlow-based deanonymization attacks against onion services.
In turn, our results provide the community with an informed assessment of the deanonymization risks posed by existing traffic telemetry equipment (i.e., sampled NetFlow), an area that has received little attention.

\para{Responsible disclosure.}
In August 2024, we shared the details of our experiments and the source code for launching introduction-flooding attacks with the Tor security team.
As a result, the team has been fully aware that we generate introduction-flooding traffic against our own onion services on the live Tor network via public relays, under the carefully controlled traffic conditions described in this section.
Since then, we have been in continuous communication with the team, and no ethical concerns have been raised to date.

Regarding mitigation efforts, the Tor security team received our first report on June 6, 2025, and acknowledged it.
The team subsequently asked several clarification questions on June 10 and 12, to which we responded on June 19.
More recently, we shared additional datasets and source code on June 12, 2026.
A shorter version of this paper was also sent to the Tor security team on June 12, while the current version was sent on September 10, before being publicly archived on October 16.
The Tor security team was informed of our plan to archive the paper on September 10 and had not raised any objection as of October 16.
\section*{Open Science}
\label{app:open-science}
Our main research artifacts consist of two components: (1) active traffic inducing attacks and (2) the passive puzzle side-channel attack.
Because unauthorized sharing or leakage of our artifacts could cause serious harm to the Tor network by enabling immediate DoS or deanonymization attacks against numerous onion services, we will not open-source them.
Instead, we selectively shared the artifacts exclusively with legitimate reviewers for the sole purpose of facilitating a thorough evaluation of our work.\footnote{For responsible disclosure, we have already shared the artifacts with the Tor security team.}



\bibliographystyle{IEEEtran}
\bibliography{main}

\appendices
\section{Formal Definitions}
\label{app:formal}
In this section, we formally define each invariant or artifact and the matching score.
Regarding notation, we denote the suggested difficulty at time $i$ as $D^{t}_{i}$ and the sampled packet count at time $i$ as $T^{t}_{i}$.
Similarly, we denote the suggested difficulty at round $i$ as $D^{r}_{i}$ and the sampled packet count at round $i$ as $T^{r}_{i}$.
Here, a {\em round} refers to a five-minute time slot between two suggested difficulty updates.
Since the adversary records the suggested difficulty levels and sampled packet counts every minute, each round consists of five difficulty values and five packet count estimates.
By definition, for round $i$ covering time slots between $j$ and $j+4$, $D^{r}_{i}$ is defined as the mode of $D^{t}_{j}, D^{t}_{j+1}, ..., D^{t}_{j+4}$, and $T^{r}_{i}$ is calculated as $\sum^{j+4}_{k=j} T^{t}_{k}$.
Furthermore, this notion of a round is used to normalize the matching score, ensuring comparability across all candidate matches.

The definitions of {\bf (Inv-I)} and {\bf (Inv-II)} are straightforward: we introduce two thresholds, $\tau_{InvI}$ and $\tau_{InvII}$.
\begin{definition}
    For two consecutive rounds $i$ and $i+1$, if $D^{r}_{i+1} > D^{r}_{i}$ and $T^{r}_{i} \geq \tau_{InvI}$, we say that rounds $i$ and $i+1$ exhibit {(\bf Inv-I)}.  
\end{definition}

\begin{definition}
    For two consecutive rounds $i$ and $i+1$, if $D^{r}_{i+1} < D^{r}_{i}$ and $T^{r}_{i} \leq \tau_{InvII}$, we say that rounds $i$ and $i+1$ exhibit {(\bf Inv-II)}.  
\end{definition}

In contrast, {\bf (Inv-III)} is defined with constants.
\begin{definition}
    For ten consecutive time slots $i$, $i+1$, ..., $i+9$, if $T^{t}_{j} \leq 10$ for all $j$ in $[0, 9]$ or $T^{t}_{j} \leq 2.5$ for all j in $[i, i+9]$ where $i \geq 10$ and $D^{t}_{i+10} \neq 0$, we say that the given trace pair violates {(\bf Inv-III)}.
\end{definition}

We define {\bf (Art-I)} in a similar manner by introducing two thresholds, $\tau_{ArtI}^{low}$ and $\tau_{ArtI}^{high}$.
These thresholds serve as the floor and the ceiling, respectively.
Additionally, since each traffic trace may contain a different number of rounds satisfying these thresholds, we also define {\bf (Art-I-b)} (base value). 
Note that $s^{p,q}_{ArtI}$ is defined as the ratio of the number of observed {\bf (Art-I)} to the number of observed {\bf (Art-I-b)}: $\frac{c^{p,q}_{ArtI}}{{c^{p,q}_{ArtIb}}}$.
If $c^{p,q}_{ArtIb} = 0$, we define $s^{p,q}_{ArtI} = 0$.
\begin{definition}
    For three consecutive rounds $i$, $i+1$, and $i+2$, if $D^{r}_{i} = D^{r}_{i+1} = D^{r}_{i+2} \neq 0$ and $\tau_{ArtI}^{low} \leq T^{r}_{k} \leq \tau_{ArtI}^{high}$ for all $k \in \{ i, i+1, i+2 \}$, we say that rounds $i$, $i+1$, and $i+2$ exhibit {(\bf Art-I)}.
\end{definition}
\begin{definition}
    For a single round $i$, if $\tau_{ArtI}^{low} \leq T^{r}_{i} \leq \tau_{ArtI}^{high}$, we say that round $i$ exhibits {(\bf Art-I-b)}.
\end{definition}

In contrast to the above definitions, we use unit time slots to define {\bf (Art-II)} instead of rounds, as seed updates occur independently of suggested difficulty updates.
\begin{definition}
    For two consecutive time slots $i$ and $i+1$, if the seed value at time $i+1$ differs from that at time $i$ and $T^{t}_{i+1} \geq 2 \times T^{t}_{i}$, we say that time slots $i$ and $i+1$ exhibit {\bf (Art-II)}.
\end{definition}

Unlike {\bf (Art-I)}, {\bf (Art-III)} requires the introduction of a threshold parameter, $\tau_{ArtIII}$, due to several key differences.
For example, {\bf (Art-III)} tend to have much shorter average durations and occur far less frequently than {\bf (Art-I)}.
As with {\bf (Art-I)}, {\bf (Art-III)} requires the introduction of a threshold parameter, $\tau_{ArtIII}$.
\begin{definition}
    For two consecutive rounds $i$ and $i+1$, if $D^{r}_{i}=D^{r}_{i+1}$, $T^{r}_{i} \leq \tau_{ArtIII}$, and $T^{r}_{i+1} \leq \tau_{ArtIII}$, we say that rounds $i$ and $i+1$ exhibit {\bf (Art-III)}.
\end{definition}

We now formally define the matching score.
\begin{definition}
    For a difficulty trace $p$ and a traffic trace $q$, the matching score is the weighted sum of the occurrences of each invariant or artifact: $MS^{p,q}=c^{p,q}_{InvI} \times w_{InvI} + c^{p,q}_{InvII} \times w_{InvII} + s^{p,q}_{ArtI} \times w_{ArtI} + c^{p,q}_{ArtII} \times w_{ArtII}$.
\end{definition}

Here, $c^{p,q}_{InvI}$, $c^{p,q}_{InvII}$, and $c^{p,q}_{ArtII}$ denote the number of observed {\bf (Inv-I)}, {\bf (Inv-II)}, and {\bf (Art-II)}, respectively.
Similarly, $s^{p,q}_{ArtI}$ is defined as the ratio of the number of observed {\bf (Art-I)} to the number of observed {\bf (Art-I-b)}: $\frac{c^{p,q}_{ArtI}}{{c^{p,q}_{ArtIb}}}$.
If $c^{p,q}_{ArtIb} = 0$, we define $s^{p,q}_{ArtI} = 0$.
The weights for each term are denoted by $w_{InvI}$, $w_{InvII}$, $w_{ArtI}$, and $w_{ArtII}$, respectively.

With noisy difficulty adjustment (\S\ref{sec:noisy}), we extend the definition of the matching score to account for the new artifact, {\bf (Art-III)}.
Let $c^{p,q}_{ArtIII}$ denote the number of {\bf (Art-III)} observed and let $w_{ArtIII}$ represent its weight.
\begin{definition}
    For a difficulty trace $p$ and a traffic trace $q$, the matching score for noisy difficulty adjustment is the weighted sum of the occurrences of each invariant or artifact: $MS^{p,q}_{noise}=c^{p,q}_{InvI} \times w_{InvI} + c^{p,q}_{InvII} \times w_{InvII} + c^{p,q}_{ArtIII} \times w_{ArtIII}$.
\end{definition}

With coarser-grained difficulty adjustment (\S\ref{sec:coarser}), we revise the definition of the matching score, as all of the artifacts are completely eliminated.
\begin{definition}
    For a difficulty trace $p$ and a traffic trace $q$, the matching score for coarser-grained difficulty adjustment is the weighted sum of the occurrences of each invariant: $MS^{p,q}_{coarser}=c^{p,q}_{InvI} \times w_{InvI} + c^{p,q}_{InvII} \times w_{InvII}$.
\end{definition}
\section{Delay Estimation Process}
\label{app:theory}
Noisy or coarser-grained puzzle difficulty updates can create discrepancies between the actual suggested difficulty level and the broadcast level.
To estimate the additional delay introduced by inaccurate (i.e., increased) suggested difficulty, we compute the difference in the expected number of hash operations required to solve a puzzle at the true difficulty versus the incorrect one.
We assume the client is using an older, non-flagship mobile device with a CPU clock speed of approximately 1.5~GHz.

\begin{definition}
    A hash operation requires a single hast step. The client can perform hash operations at a rate of $\gamma$ steps per second.
\end{definition}

\begin{definition}
    A hash function is modeled as a random oracle. It behaves like a random function, but produces a consistent output for any given input.
\end{definition}

\begin{definition}
    A puzzle of difficulty $D$ requires finding a value $R$ such that $R \times D \leq 2^{32}$, where $R$ is the output of a hash operation. Both $R$ and $D$ are non-negative integers, with $R \in [0, 2^{32}-1]$ and $D \in [1, 2^{32}-1]$.
\end{definition}

\begin{lemma}
\label{lem:single-step-prob}
    The probability of successfully solving a puzzle of difficulty $D$ with a single hash step is given by
    \[
        \frac{\bigl \lfloor \frac{2^{32}-1}{D} \bigr \rfloor + 1}{2^{32}}.
    \]
\end{lemma}

\begin{proof}
    The range of possible hash outputs is $[0, 2^{32}-1]$.
    A solution space satisfies $R \times D \leq 2^{32}$, which implies $R \leq \lfloor \frac{2^{32}-1}{D} \rfloor$.
    Thus, the number of valid $R$ values is $\lfloor \frac{2^{32}-1}{D} \rfloor + 1$.
    Since each $R$ is equally likely, the probability of success in a step is the ratio of valid outputs to total possible outputs.
\end{proof}

\begin{lemma}
\label{lem:single-step-prob-fail}
    The probability of failing to solve a puzzle of difficulty $D$ within 
    \[
        \Biggl \lceil \frac{-2}{\log{\bigl(1-{ \frac{{\lfloor \frac{2^{32}-1}{D}} \rfloor + 1}{2^{32}}\bigr)}}} \Biggr \rceil
    \]
    hash steps is less than or equal to $10^{-2}$.
\end{lemma}

\begin{proof}
    Let $X$ be the random variable representing the number of puzzles being solved in $n$ hash steps.
    Since each step is an independent Bernoulli trial with success probability
    \[
        p = \frac{\bigl \lfloor \frac{2^{32}-1}{D} \bigr \rfloor + 1}{2^{32}},
    \]
    $X$ follows a binomial distribution
    \[
        P(X=k) = \binom{n}{k} p^{k} (1-p)^{n-k}, \quad \text{for } k \in \mathbb{Z}^+_0.
    \]

    To ensure that the probability of no successes (i.e., $P(X=0)$) is at most $10^{-2}$, we require 
    \[
        (1-p)^n \leq 10^{-2}.
    \]

    Taking logarithms on both sides,
    \[
        n \log{(1-p)} \leq -2,
    \]
    \[
        n \geq \frac{-2}{\log{(1-p)}}.
    \]

    When $D\neq1$, since $n$ must be an integer, taking the ceiling,
    \[
        n \geq \Biggl \lceil \frac{-2}{\log{\bigl(1-{ \frac{{\lfloor \frac{2^{32}-1}{D}} \rfloor + 1}{2^{32}}\bigr)}}} \Biggr \rceil.
    \]

    For the special case where $D=1$, we have $p=1$, so $P(X=0)=0$, which trivially satisfies the condition $P(X=0) \leq 10^{-2}$.
\end{proof}

Now, we estimate the additional delay incurred due to the inaccurate difficulty using the previously defined concepts and results.
Rather than directly calculating the difference in the number of hash steps required to solve puzzles at the accurate and inaccurate difficulty levels, we approximate this by considering the number of steps required to solve a separate puzzle whose difficulty equals the difference between two levels.
We manually verify that the maximum error introduced by this approximation is at most three hash steps when the inaccuracy in the difficulty level exceeds 1.
In the case where the inaccuracy is exactly 1, the error can be up to five hash steps.

\begin{theorem}
\label{theom:single-step}
    Let the inaccurate difficulty exceed the true difficulty level by $\delta$, where $\delta > 1$.
    Then, the expected additional delay $T_{avg}$ incurred due to this inaccuracy is $\frac{\delta}{\gamma}$ seconds.
    Furthermore, with probability at most 0.01, the worst-case additional delay $T_{worst}$ can be as large as 
    \[
        \frac{\biggl \lceil \frac{-2}{\log{\bigl(1-{ \frac{{\lfloor \frac{2^{32}-1}{\delta}} \rfloor + 1}{2^{32}}\bigr)}}} \biggr \rceil}{\gamma}.
    \]
\end{theorem}

\begin{proof}
    By Lemma~\ref{lem:single-step-prob}, the expected number of hash steps required to solve a puzzle of difficulty $D$ is approximately $D$.\footnote{This holds when $D$ is not too large, as in Tor where $D$ is practically capped at 10,000.}
    Therefore, if the difficulty is overestimated by $\delta$, the expected number of additional hash steps is $\delta$.

    From Lemma~\ref{lem:single-step-prob-fail}, to ensure that the probability of failing to solve a puzzle of difficulty $\delta$ within $n$ hash steps is at most 0.01, the number of steps must satisfy
    \[
        n \geq \Biggl \lceil \frac{-2}{\log{\bigl(1-{ \frac{{\lfloor \frac{2^{32}-1}{\delta}} \rfloor + 1}{2^{32}}\bigr)}}} \Biggr \rceil.
    \]
    This gives the worst-case number of hash steps needed with 99\% confidence.

    Since the client performs $\gamma$ hash steps per second, the expected delay is given by
    \[
        T_{avg} = \frac{\delta}{\gamma},
    \]
    and the worst-case delay (with 1\% probability of being exceeded) is bounded by
    \[
        T_{worst} = \frac{\biggl \lceil \frac{-2}{\log{\bigl(1-{ \frac{{\lfloor \frac{2^{32}-1}{\delta}} \rfloor + 1}{2^{32}}\bigr)}}} \biggr \rceil}{\gamma}.
    \]
\end{proof}

Finally, given that the assumed low-performance device can perform at least 100 hash operations per second, the additional delay can be expressed as $\frac{\delta}{100}$ seconds on average.
In the worst-case scenario, with a probability of at most 1\%, the delay can be as high as
\[
    \frac{\biggl \lceil \frac{-2}{\log{\bigl(1-{ \frac{{\lfloor \frac{2^{32}-1}{\delta}} \rfloor + 1}{2^{32}}\bigr)}}} \biggr \rceil}{100}.
\]

Note that modern laptops or desktops can perform hash operations much faster; e.g., Intel\textsuperscript{\textregistered} Xeon\textsuperscript{\textregistered} Silver 4214R Processor can achieve over 250 hash operations per second, while Intel\textsuperscript{\textregistered} Core\textsuperscript{\textregistered} i7-1260P Processor can reach approximately 300 operations per second.

\end{document}